\documentclass[twoside,leqno,twocolumn]{article}

\usepackage[letterpaper]{geometry}

\usepackage{ltexpprt}

\usepackage{url}
\usepackage{amsmath}
\usepackage{algorithmic}
\usepackage{amssymb}
\usepackage{color}
\usepackage{array}
\usepackage{xy}
\usepackage{setspace}
\usepackage{multicol}
\usepackage{hyperref}
\usepackage{cite}

\usepackage{accents}

\newcommand{\defn}[1]{\emph{\textbf{{#1}}}}
\newcommand{\av}{\operatorname{av}}
\newcommand{\tot}{\operatorname{tot}}
\usepackage{todonotes}

\newcommand{\poly}{\operatorname{poly}}
\newcommand{\polylog}{\operatorname{polylog}}
\newcommand{\E}{\mathbb{E}}

\begin{document}

\title{Achieving Optimal Backlog in the Vanilla Multi-Processor Cup Game}
\author{William Kuszmaul\thanks{MIT CSAIL. Supported by NSF Grants CCF 1314547 and CCF 1533644; by a  Fannie \& John Hertz Foundation Fellowship; and by an NSF GRFP Fellowship.}}
\date{}
\maketitle

\fancyfoot[R]{\scriptsize{Copyright \textcopyright\ 2020 by SIAM\\
Unauthorized reproduction of this article is prohibited}}

\begin{abstract}\small\baselineskip=9pt
In each step of the $p$-processor cup game on $n$ cups, a filler distributes up to $p$ units of water among the cups, subject only to the constraint that no cup receives more than $1$ unit of water; an emptier then removes up to $1$ unit of water from each of $p$ cups. Designing strategies for the emptier that minimize backlog (i.e., the height of the fullest cup) is important for applications in processor scheduling, buffer management in networks, quality of service guarantees, and deamortization.

We prove that the greedy algorithm (i.e., the empty-from-fullest-cups algorithm) achieves backlog $O(\log n)$ for any $p \ge 1$. This resolves a long-standing open problem for $p > 1$, and is asymptotically optimal as long as $n \ge 2p$.

If the filler is an oblivious adversary, then we prove that there is a randomized emptying algorithm that achieve backlog $O(\log p + \log \log n)$ with probability $1 - 2^{-\polylog(n)}$ for $2^{\polylog(n)}$ steps. This is known to be asymptotically optimal when $n$ is sufficiently large relative to $p$. The analysis of the randomized algorithm can also be reinterpreted as a smoothed analysis of the deterministic greedy algorithm.

Previously, the only known bound on backlog for $p > 1$, and the only
known randomized guarantees for any $p$ (including when $p = 1$),
required the use of resource augmentation, meaning that the filler can
only distribute at most $p(1 - \epsilon)$ units of water in each step,
and that the emptier is then permitted to remove $1 + \delta$ units of
water from each of $p$ cups, for some $\epsilon, \delta > 0$.
\end{abstract}

\section{Introduction}

The \defn{$p$-processor cup game} is a multi-round game in which there
are $n$ cups, each initially empty, and two players take turn placing
water into and removing water from the cups. In each step of the game,
the \defn{filler} distributes (up to) $p$ units of water among $n$
cups, placing no more than $1$ unit in any individual cup; the
\defn{emptier} then selects $p$ cups and removes (up to) $1$ unit of
water from each. The goal of the emptier is to minimize the amount of
water in the fullest cup, also known as the \defn{backlog} of the
system. When $p = 1$, the game is often called the
\defn{single-processor cup game}, and when $p > 1$, the game is often
called the \defn{multi-processor cup game}.

The $p$-processor cup game naturally arises in the study of processor
scheduling \cite{AdlerBeFr03,LitmanMo09,DietzRa91}. The $n$ cups
represent tasks that are being scheduled to run in time slices on $p$
processors. The amount of water in each cup represents work that needs
to be performed on that task. During each time slice, the $p$
processors can each make $1$ unit of progress on some task, although
no two processors can work on the same task. New work also arrives
during each time slice, with $p$ new units of work being distributed
arbitrarily among the tasks, subject only to the constraint that each
task receives a maximum of $1$ new unit of work. (Note that without
this constraint, all of the new work might go to a single task, in
which case the processors would have no possibility of keeping up.)
The goal is to design a scheduling algorithm that prevents any task
from having a large amount of unfinished work. This corresponds directly
to finding an emptying algorithm for the $p$-processor cup game that
achieves small backlog.

The $p$-processor cup game (and its relaxations) has also appeared in
a variety of other applications, ranging from deamortization
\cite{AmirFaId95,DietzRa91,DietzSl87,AmirFr14,Mortensen03,GoodrichPa13,FischerGa15,Kopelowitz12},
to buffer management in network
switches~\cite{Goldwasser10,AzarLi06,RosenblumGoTa04,Gail93}, to
quality of service
guarantees~\cite{BaruahCoPl96,AdlerBeFr03,LitmanMo09}.

A natural emptying strategy is \defn{the greedy algorithm}, in which
the emptier always selects the $p$ fullest cups at each step. For the
single-processor cup game, the greedy algorithm is known to achieve
backlog $O(\log n)$, which is optimal for any deterministic algorithm
\cite{AdlerBeFr03, DietzSl87}. Designing algorithms with provable guarantees for the full multi-processor cup game has proven more
difficult, and has remained an open question since the late 1960's \cite{Liu69}. Even when $p = 2$, no nontrivial bounds on the
performance of the greedy algorithm are known.

The difficulty of analyzing the multi-processor cup game stems from
the constraint that the emptier can remove at most $1$ unit of water
from each cup. This means that, even when a single cup contains far
more water than any other cup, only a $\frac{1}{p}$-fraction of the
emptier's resources can be devoted to emptying that cup. As was first
noted by Liu~\cite{Liu69} in 1969, and later reiterated by other
authors~\cite{LitmanMo09,BaruahCoPl96, GuanYi12, BenderFaKu19}, this
constraint adds a ``surprising amount of difficulty'' to the
scheduling problem.

\vspace{.3 cm} \noindent \textbf{Resource augmentation to help the
  emptier. }In order to analyze the multi-processor cup game, a key
tool has been the use of resource augmentation. In the
\defn{$(\epsilon, \delta)$-resource-augmented $p$-processor cup game},
the filler is restricted at each step to distribute at most
$p(1-\epsilon)$ units of water among the cups, placing no more than
$1$ unit in each cup, and then the emptier is permitted to remove up
to $1 + \delta$ water from each of $p$ cups.

Recently, Bender et al. \cite{BenderFaKu19} showed that as long as
$\delta \ge \frac{1}{\poly(n)}$, the greedy emptying algorithm
achieves backlog $O(\frac{1}{\epsilon} \log n)$. When $\epsilon \ge
\Omega(1)$, this results in an asymptotically optimal backlog of
$O(\log n)$.

Resource augmentation has also played a pivotal role in the design of
\emph{randomized} emptying algorithms for both the single- and multi-
processor cup games. When the emptier's algorithm is randomized, the
filler is presumed to be an oblivious adversary.

Dietz and Raman~\cite{DietzRa91} proved that if the emptier is
permitted to completely empty $p$ cups (i.e., $\delta = \infty$),
then there is a randomized emptying algorithm that achieves backlog
$O(\log \log n)$ with probability $1 - 1/\poly(n)$. Dietz and Raman
also give a matching lower bound construction that achieves backlog
$\Omega(\log \log n)$ with probability $\frac{1}{\poly(n)}$.

Recently, Bender et al. \cite{BenderFaKu19} presented several
randomized algorithms that work with far less resource
augmentation. For the single-processor game, they introduced the
\defn{smoothed greedy algorithm}, which achieves backlog $O(\log \log
n)$ with probability $1 - \frac{1}{\poly(n)}$, as long as $\epsilon
\ge \frac{1}{\polylog(n)}$. 

Bender et al. \cite{BenderFaKu19} also gave a more intricate algorithm
for the multi-processor cup game. Their algorithm achieves backlog
$O(\frac{1}{\epsilon} \log \log n)$ with probability $1 -
\frac{1}{\poly(n)}$, as long as $\epsilon, \delta \ge
\frac{1}{\poly(p)}$ satisfy certain natural constraints. The algorithm
has the additional remarkable property that it achieves backlog $O(1)$
after a given step with probability $1 - e^{-\Omega(\epsilon^2 p)}$.

The algorithms of Dietz and Raman~\cite{DietzRa91} and Bender et
al.~\cite{BenderFaKu19} both rely heavily on the use of resource
augmentation to achieve their bounds. Designing randomized algorithms
for either the single-processor or multi-processor cup game that do
not rely on resource augmentation remains an open question.

\vspace{.3 cm} \noindent \textbf{Our results: augmentation-free analyses of greedy
  and smoothed-greedy. } We prove that the greedy algorithm for the
$p$-processor cup game on $n$ cups achieves backlog $O(\log n)$
deterministically. Moreover, we show that no deterministic algorithm
can do asymptotically better, as long as $n \ge 2p$. At the heart of
our analysis is an intricate system of invariants for the
$p$-processor cup game.

Our second main result is an analysis of the smoothed greedy algorithm
that does not rely on resource augmentation, and that works for any
number of processors $p$. We show that the algorithm achieves backlog
$O(\log p + \log \log n)$ with probability $1 - 2^{-\polylog(n)}$ for
all of the first $2^{\polylog(n)}$ steps of the $p$-processor cup game. For
fixed $p$, and as $n$ grows large, this becomes backlog $O(\log \log
n)$, which is known to be asymptotically optimal.

As noted by Bender et al. \cite{BenderFaKu19}, any analysis of the
smoothed greedy algorithm also doubles as a \emph{smoothed analysis}
for the deterministic greedy algorithm. This is because the smoothed
greedy algorithm works by first randomly perturbing the initial
starting state of the game, and then following a variant of the
standard deterministic greedy algorithm on the perturbed game.

\vspace{.3 cm} \noindent \textbf{Lower bounds against possible improvements. }We also
prove several impossibility results separating the $p$-processor cup
game from its resource-augmented counterpart.

First, we design an oblivious strategy for the filler that achieves
backlog $\Omega(\log \log p)$ with \emph{constant} probability. This
contrasts with resource-augmented $p$-processor cup game, for which a
randomized emptying algorithm is known to achieve backlog $O(1)$ with
probability $1 - e^{- \Omega(\epsilon^2 p)}$ \cite{BenderFaKu19}.

The second lower bound considers the possibility of a randomized
emptying algorithm for the multi-processor cup game that provides an
\defn{unending guarantee}, meaning that for any (arbitrarily large)
time step $t \in \mathbb{N}$, the algorithm gives an
at-least-constant-probability bound of $o(\log n)$ on the backlog at
time $t$. The randomized algorithms of Bender et
al. \cite{BenderFaKu19} all provide unending guarantees with the use
of resource augmentation. We show that, without resource augmentation,
unending guarantees are impossible for any ``greedy-like'' emptying
strategy, including the known variants of the smoothed greedy
algorithm.

\vspace{.3 cm} \noindent \textbf{Related work and problem history. }The problem of
designing and analyzing low-backlog algorithms for the $p$-processor
cup game was first discussed in 1969 by Liu \cite{Liu69}.  Much of the
work on the problem (and especially early work on the problem) adds
the additional constraint that water arrives at \emph{fixed rates},
meaning that each cup $j$ receives the same fixed amount of water
$f_j$ during each step, with $\sum_j f_j = p$
\cite{BaruahCoPl96,GkasieniecKl17,BaruahGe95,LitmanMo11,LitmanMo05,MoirRa99,BarNi02,GuanYi12,Liu69,
  LiuLa73}.

In 1987, Dietz and Sleator analyzed the greedy algorithm
\cite{DietzSl87} for a relaxed version of the cup game in which the
emptier is permitted to remove all of the water from a cup -- this is
sometimes called the \defn{cup flushing game}. Adler et
al. \cite{AdlerBeFr03} later improved upon this to achieve a bound of
$O(\log n)$ for single-processor cup game without the use of resource
augmentation.

Initial efforts to extend the result of Adler et al. to $p > 1$
required either that the emptier be able to remove multiple units of
water from a single cup \cite{AdlerBeFr03} (in which case the main task at hand becomes to prove lower bound constructions), or that the emptier be
able to see (at least partially) into the future \cite{LitmanMo09,
  GuanYi12}. In recent work, Bender et al. \cite{BenderFaKu19} showed
how to bound the backlog using only resource augmentation. Presenting
a bound without the use of resource augmentation (or
semi-clairvoyance) has until now remained open.

The effort to design randomized algorithm that achieve better backlog
than their deterministic counterparts began with work by Dietz and Raman in 1991 \cite{DietzRa91} that
achieved a bound of $O(\log \log n)$ for the cup flushing game. Dietz
and Raman posed as an open question whether a simpler randomized
algorithm might exist. The smoothed greedy algorithm later introduced
by Bender et al. \cite{BenderFaKu19} gives one solution to this
problem, and also allows for backlog guarantees to be proven in the
presense of only a small amount of resource augmentation. This raised
the question of whether backlog bounds could be achieved without the
use of any resource augmentation, which is addressed in this paper.

In the context of packet-switching, Bar-Noy et
al.~\cite{Bar-NoyFrLa02} considered a variant of the single-processor
cup game in which the filler is permitted to place arbitrarily large
amounts of water into the cups at each step, subject only to the
constraint that each cup receives an integer amount. Rather than
proving absolute bounds on backlog, which would be impossible, the
authors prove that the greedy algorithm is $O(\log n)$-competitive
with the offline optimal algorithm. Moreover, they show that no online
algorithm, including randomized algorithms, can do better. Several
similar results were also discovered concurrently by Fleischer and
Koga~\cite{FleischerKo04}. Subsequent work has considered weaker
adversaries~\cite{DamaschkeZh05}.

Researchers have also studied variants of the multi-processor cup game
in which the cups are nodes on a graph, and in which further
constraints are placed on the emptier based on the structure of the
graph~\cite{BodlaenderHuKu12,BenderFeKr15,BodlaenderHuWo11,ChrobakCsIm01}. This
can be used to model multiprocessor scheduling with conflicts between
tasks~\cite{BodlaenderHuWo11,ChrobakCsIm01} as well as some problems
in sensor radio networks~\cite{BenderFeKr15}.

\vspace{.3 cm} \noindent \textbf{Paper outline. }The rest of the paper
proceeds as follows. In Section \ref{sec:prelims} we establish
conventions and notations for discussing the $p$-processor cup
game. In Section \ref{sec:technical} we give a technical overview of
the main results in the paper. In Section \ref{sec:greedy}, we present
the full analysis of the greedy algorithm. Finally, in Section
\ref{sec:smoothedgreedy}, we analyze the smoothed greedy algorithm,
and prove lower bounds against constant backlog and unending
guarantees.

\section{Preliminaries}\label{sec:prelims}

In this section, we briefly establish conventions and notation for
discussing the $p$-processor cup game on $n$ cups.  A \defn{cup state}
$S$ assigns a non-negative fill to each of the cups $1, \ldots,
n$. The fill of the $i$-th fullest cup, also known as the
\defn{rank-$i$ cup}, in a state $S$ is denoted by $S(i)$. Note that
$S(i)$ may differ from the fill of cup $i$ (i.e., the cup with label
$i$). The total fill of the $i$ fullest cups is denoted by
$\tot_i(S)$, and the average fill of the $i$ fullest cups is denoted
by $\av_i(S)$. If $X \subset [n]$ is a set of cup-indices, then we
define $\av_X(S)$ to denote the average fill of the cups $X$ in state
$S$; and similarly we define $\tot_X(S)$ to be the total fill of the
cups $X$ in state $S$.

We use $S_t$ to denote the cup state after the $t$-th step in the
game, with $S_0$ representing the initial (empty) state, and we use
$\mathcal{S}_t$ to denote the set of all the possible states that
$S_t$ could take.

\section{Technical Overview}\label{sec:technical}

In this section, we present technical overviews of the paper's main
results. Section \ref{sec:greedy_overview} discusses the analysis of
the greedy algorithm for the $p$-processor cup game. Section
\ref{sec:randomized_overview} then discusses the analysis of the
(randomized) smoothed greedy algorithm. Finally Section
\ref{sec:lowerbounds_overview} briefly describes lower-bound
constructions for the filler.

\subsection{Analyzing Greedy}\label{sec:greedy_overview}

Theorem \ref{thm:lognbacklog_pre} establishes that if the emptier in the
$p$-processor cup game follows the greedy strategy, then the backlog
never exceeds $O(\log n)$, and that this is tight for $n \ge 2p$.

\begin{theorem}
The greedy algorithm for the $p$-processor cup game on $n$ cups
achieves backlog $O(\log n)$. Moreover, there is an adaptive strategy
for the filler that achieves backlog $\Omega(\log (n - p))$ against
any emptying strategy.
\label{thm:lognbacklog_pre}
\end{theorem}

Here, we focus on proving the upper bound of $O(\log n)$, and we defer
the discussion of the lower bound to Subsection
\ref{sec:lowerbounds_overview}.

\vspace{.3 cm}

\noindent \textbf{Proving a bound of $O(p + \log n)$. } We begin by
proving a weaker bound of $O(p + \log n)$ on the backlog.

Define the \defn{$N$-truncated} $p$-processor cup game to be the
standard $p$-processor cup game, except with one additional
restriction on the filler: the filler is never permitted to increase
the fill of any cup to be larger than $N$. We prove that, if $N$ is a
sufficiently large constant multiple of $p + \log n$, then the greedy
emptying algorithm prevents the backlog from ever exceeding $N - 2$;
it follows that the maximum backlog in the standard (non-truncated)
game is also at most $N - 2 \le O(p + \log n)$.

The key to analyzing the $N$-truncated cup game is to examine the
\defn{$N$-skewed averages} of cup states. These are defined by,
\begin{equation}
f^N_k(S) = \max\left(\frac{\tot_{p + k}(S) - p \cdot N}{k}, 0\right),
\label{eq:fk}
\end{equation}
for $k = 1, 2, \ldots, n - p$. Combinatorially, one should think of
the $N$-skewed average $f^N_k(S)$ of a state $S$ as follows: Take the
total amount of water in the $p + k$ fullest cups; move as much of
that water to the $p$ fullest cups as possible, without allowing the
average fill of the $p$ fullest cups to exceed $N$; then distribute the
remaining water evenly among the cups of rank $p + 1, \ldots, p +
k$. The amount of water in each of the cups of rank $p + 1, \ldots, p
+ k$ is the skewed average $f^N_k(S)$.

The main technical step in the proof is to prove a system of invariants
involving the $n$-skewed averages.
\begin{lemma}
Supposed the emptier follows the greedy algorithm for the the
$N$-truncated $p$-processor cup game on $n$ cups. Then, for all steps
$t \ge 0$ and for all $k \in [n - p]$,
$$f^N_k(S_t) \le 1 + \left(\frac{1}{k + 1} + \frac{1}{k + 2} + \cdots + \frac{1}{n}\right).$$
\label{lem:truncatedinvariant_technical_overview}
\end{lemma}

Using the case of $k = 1$, Lemma
\ref{lem:truncatedinvariant_technical_overview} implies that
$f^N_1(S_t) \le O(\log n)$, and thus that,
\begin{equation}
  \tot_{p + 1}(S_t) \le O(\log n) + p \cdot N,
  \label{eq:fencepost}
\end{equation}
for all steps $t \ge 0$.

To exploit \eqref{eq:fencepost}, we prove the useful combinatorial
fact that, whenever $S_t(1) > S_{t - 1}(1)$ for consecutive steps $t
-1$ and $t$, the old state $S_{t - 1}$ necessarily satisfies $\tot_{p
  + 1}(S_{t - 1}) \ge (p + 1) \cdot (S_t(1) - 1)$. It follows that if
any state $S_t$ ever achieves backlog $N - 2$ or greater, then
$$\tot_{p + 1}(S_{t - 1}) \ge (p + 1) \cdot (N - 3) \ge (N - 3 - 3p) +
p \cdot N,$$ which contradicts \eqref{eq:fencepost} when $N$ is a
sufficiently large constant multiple of $p + \log n$. Assuming Lemma
\ref{lem:truncatedinvariant_technical_overview}, it follows that
backlog can never exceed $O(p + \log n)$.

\begin{proof}[Proof Sketch for Lemma \ref{lem:truncatedinvariant_technical_overview}]
We prove the result by induction on $t$. Consider $t > 0$, and suppose
as an inductive hypothesis that the result holds for $t - 1$. 

Let $I_t$ denote the intermediate state between $S_{t- 1}$ and $S_t$
after the filler has placed water into cups, but the emptier has not
yet removed water. To get from state $I_t$ to state $S_t$, one unit of
water is removed from each of the $p$ fullest cups\footnote{For
  simplicity, we are ignoring the case here in which one or more of
  the $p$ fullest cups contains less than $1$ unit of water in state
  $I_t$.}. Of these $p$ cups, let $A$ denote the subset that remain
among the $p + k$ fullest cups in $S_t$, and let $B$ denote the subset
that do not.

Let $I'_t$ denote an intermediate state between $I_t$ and $S_t$ at
which one unit of water has been removed from each cup in $A$, but not
from each cup in $B$. The $p + k$ fullest cups in $S_t$ are the $p + k
+ |B|$ fullest cups in $I'_t$, except without the cups in $B$, meaning
that,
\begin{equation}
  \tot_{p + k}(S_t) = \tot_{p + k + |B|}(I'_t) - \tot_{B}(I'_t).
  \label{eq:totB}
\end{equation}

Rewriting this in terms of $N$-skewed averages yields,
\begin{equation}
  k \cdot f^N_{k}(S_t) = (k + |B|) \cdot f^N_{k + |B|}(I'_t) - \tot_{B}(I'_t),
  \label{eq:tottoav}
\end{equation}
unless one of $f^N_k(S_t)$ or $f^N_{k + |B|}(I'_t)$ is zero. A key
property of the $N$-skewed average $f^N_{k + |B|}(I'_t)$ is that, because of
the $N$-truncation requirement,
\begin{align*}
  f^N_{k + |B|}(I'_t) &= \frac{\tot_{p + k + |B|}(I'_t) - p \cdot N}{k + |B|} \\ &\le
\frac{\tot_{p + k + |B|}(I'_t) - \tot_{p}(I'_t)}{k + |B|},\end{align*} unless $f^N_k(I'_t) =
0$. The right-hand side can be combinatorially interpreted as the
average fill of the cups with ranks $p + 1, p + 2, \ldots, p + k + |B|$ in
state $I'_t$. Since the cups in $B$ have ranks in $[p]$, each cup in
$B$ therefore has fill at least $f^N_{k + |B|}(I'_t)$. Applying this to
\eqref{eq:tottoav},
\begin{equation*}
  k \cdot f^N_{k}(S_t) = (k + |B|) \cdot f^N_{k + |B|}(I'_t) - |B| \cdot f^N_{k + |B|}(I'_t),
\end{equation*}
meaning that $f^N_k(S_t) \le f^N_{k + |B|}(I'_t)$.

The $p + k + |B|$ fullest cups in state $I'_t$ contain, in total, at
most $p - |A| = |B|$ more water in state $I'_t$ than the $p + k + |B|$
fullest cups contained in state $S_{t - 1}$. Thus we have,
$$f^N_k(S_t) \le f^N_{k + |B|}(I'_t) \le f^N_{k + |B|}(S_{t - 1}) + \frac{|B|}{k +
  |B|},$$
which by the inductive hypothesis is at most,
\begin{equation*}
  \begin{split}
    & 1 + \left(\frac{1}{k + |B| + 1} + \frac{1}{k + |B| + 2} + \cdots + \frac{1}{n - p}\right) + \frac{|B|}{k + |B|} \\ & \le 1 + \left(\frac{1}{k + 1} + \frac{1}{k + 2} + \cdots + \frac{1}{n - p}\right).
  \end{split}
\end{equation*}
\end{proof}

\vspace{.3 cm}

\noindent \textbf{Proving the full bound of $O(\log n)$. }The analysis
described above requires two major changes in order to achieve the
stronger bound of $O(\log n)$.

The first change is to analyze the $M$-skewed average of the
\emph{non-truncated} cup game, where
\begin{equation*}
M = \sup_{t \in \mathbb{N}, \, S_t \in \mathcal{S}_t} \av_p(S_t).
\end{equation*}
Since we have already proven that no cup ever has fill greater than
$O(\log n + p)$, we know that $M$ is finite.

Proving a result analogous to Lemma
\ref{lem:truncatedinvariant_technical_overview} is made difficult by
the fact that the cup game is not necessarily $M$-truncated. Moreover,
if we define $I'_t$ as in the proof of Lemma
\ref{lem:truncatedinvariant_technical_overview}, then $\tot_p(I'_t)$
could potentially be as large as $p \cdot M + |B|$ (if each cup $B$
contains $M + 1$ units of water, and each cup not in $B$ contains $M$
units of water); the proof of the lemma, on the other hand, requires $\tot_p(I'_t) \le p \cdot M$. Nonetheless, by exploiting the combinatorial properties
of the quantity $M$, we can establish that for all $1 \le k \le n -
p$, and for all states $S_t$,
\begin{equation}
  f^M_k(S_t) \le 1 + \left(\frac{1}{k + 1} + \frac{1}{k + 2} + \cdots + \frac{1}{n}\right).
  \label{eq:fbound}
\end{equation}

The second change to the analysis is in the use of
\eqref{eq:fbound}. We show that, in order to achieve a bound on
backlog of $O(\log n)$, it suffices to bound $M$ by $O(\log n)$. Thus
the challenge becomes to use \eqref{eq:fbound} in order to bound $M
\le O(\log n)$.

By \eqref{eq:fbound}, with $k = 1$, we know that
\begin{equation}
  \tot_{p + 1}(S_t) \le p \cdot M + O(\log n),
  \label{eq:fkequals1}
\end{equation}
for all steps $t$. We will show that, in order for this to always be
true, we must have $M \le O(\log n)$.

Call a step $t$ \defn{record-setting} if $\tot_{p}(S_t) > \tot_{p}(S_{t'})$ for all $t' < t$. The key to completing the proof is
to establish that, whenever a step $t$ is record-setting, the fills
$S_t(1), S_t(2), \ldots, S_t(p + 1)$ are all very close to one
another. In particular, we establish that,
\begin{equation}
  S_t(1) - S_t(p + 1) \le O(\log n).
  \label{eq:closetogethercups}
\end{equation}

By the definition of $M$, for any $\epsilon > 0$, there must exist a
sequence of step states $S_1, \ldots, S_t$ such that $S_t$ is a
record-setting state that achieves $\tot_p(S_t) = p \cdot M -
\epsilon$. It follows by \eqref{eq:closetogethercups} that $\tot_{p +
  1}(S_t) \ge (p + 1) \cdot M - O(\log n)$. In order to avoid a
contradiction with \eqref{eq:fkequals1}, $M$ is therefore forced to be at most
$O(\log n)$.

The final component to the proof is showing
\eqref{eq:closetogethercups}. Using the fact that $\tot_{p + 1}(S_t) >
\tot_{p + 1}(S_{t'})$ for all $t' < t$, we argue that for all $j \in [p]$,
\begin{equation}
  \frac{S_t(j + 1) + \cdots + S_t(p + 1)}{p + 1 - j} \ge S_t(j) - 1.
  \label{eq:singleconstraint}
\end{equation}
In
particular, this is because whichever step $t' \le t$ is the first
step to have had $S_{t'}(1) \ge S_{t}(1), S_{t'}(2) \ge S_{t}(2), \ldots, S_{t'}(j) \ge S_{t}(j)$
can be also seen to satisfy,
$$S_{t'}(j + 1), \ldots, S_{t'}(p + 1) \ge S_{t'}(j) - 1.$$ In order
so that $\tot_{p + 1}(S_t) \ge \tot_{p + 1}(S_{t'})$ (that way $t$ can
be a record-setting step), it follows that \eqref{eq:singleconstraint}
must hold.

The inequality \eqref{eq:singleconstraint} induces a system of
constraints on the values $S_t(1), S_t(2), \ldots, S_t(p + 1)$, which
together can be shown to enforce \eqref{eq:closetogethercups},
completing the analysis.

\subsection{Analyzing Smoothed Greedy}\label{sec:randomized_overview}

In Section \ref{sec:smoothedgreedy}, we analyze the \defn{smoothed
  greedy algorithm} for the $p$-processor cup game on $n$ cups
\cite{BenderFaKu19}. At the beginning of algorithm, the emptier
selects independent values $r_j$ uniformly at random from the interval
$[0, 1]$ for each $j = 1, \ldots, n$. Prior to the first step of the
game, the emptier inserts $r_j$ water into each cup $j$. The emptier's
strategy at the end of each step is then to simply remove $1$ unit of
water from each of the $p$ fullest cups. If, however, one or more of
the $p$ fullest cups contains less than $1$ unit of water, then the
emptier does not remove from those cups. This ensures the important
property that the fractional amount of water (i.e., the amount of
water modulo $1$) in each cup $j$ after a step $t$, is a function only
of the initial random offset $r_j$ and the filler's actions in the
first $t$ steps.

It was shown by \cite{BenderFaKu19} that, if the emptier has $\epsilon \ge
\frac{1}{\polylog n}$ resource augmentation, then with high
probability in $n$ the smoothed greedy algorithm achieves backlog
$O(\log \log n)$ in each step of the single-processor cup game.

We present a new analysis of the smoothed greedy algorithm that
applies without the use of resource augmentation, and to $p >
1$. Prior to this result, no sub-logarithmic bounds were known for any
randomized algorithm without resource augmentation, even in the case
of $p = 1$.

\begin{theorem}
  Let $c$ be at least a sufficiently large constant. Then with
  probability at least $1 - \exp\left(-\log^c
  n\right)$, the smoothed greedy algorithm achieves
  backlog $O(\log p + c \log \log n)$ after all of the first
  $\exp(\log^c n)$ steps.
  \label{thm:randomized_pre}
\end{theorem}

The proof of Theorem \ref{thm:randomized_pre} analyzes $\Theta(\log
\log n)$ cup games concurrently, where the participant cups in the
level-$i$ cup game, also known as the \defn{level-$i$ active cups},
are the cups containing fill greater or equal to $2(i - 1)$. We denote
the \defn{number of level-$i$ active cups after step $t$ by
  $A^{(i)}(t)$}. If $m_i$ is the maximum number of level-$i$ active
cups during any of the first $2^{\polylog(n)}$ steps, then we wish to
show that $m_{i + 1} \le m_i^{.75}$ (unless $m_i$ is already quite
small).

Define $h_j(t)$ to be the height (i.e. fill) of cup $j$ after step
$t$. Define the \defn{level-$i$ fill} $h^{(i)}_j(t)$ of cup $j$ after
step $t$ to be $\max(h_j(t) - 2 (i - 1), 0)$. The set of level-$(i +
1)$ active cups after step $t$ is precisely the set of cups for which
$h^{(i)}_j(t) > 2$.

Rather than bounding the number of level-$(i + 1)$ active cups
directly, we instead bound a larger quantity that we call the
level-$i$ integer fill. The \defn{level-$i$ integer fill} after step
$t$, denoted by $T^{(i)}(t)$, is given by
$$T^{(i)}(t) = \sum_{j = 1}^n \max\left(\lfloor h^{(i)}_j(t) - 1
\rfloor, 0 \right).$$ We say a cup \defn{crosses a level-$i$
  threshold} $s$ at step $t$, whenever the amount of water $f$ placed
by the filler into cup $j$ during step $t$ satisfies $h^{(i)}_j(t - 1)
< s \le h^{(i)}_j(t - 1) + f$ for some integer $s \ge 2$. One can
think of the level-$i$ integer fill as counting the number of
level-$i$ threshold crossings that have not yet been undone by the
emptier.

The level-$i$ integer fill is at least as large as the number of
level-$(i + 1)$ active cups, since each active cup contributes at
least $1$ to the integer fill. One of the difficulties of bounding the
number of level-$(i + 1)$ active cups directly is that there may be
long series of steps during which the emptier is focused on very full
cups, and during which the filler is able to increase the number of
level-$(i + 1)$ active cups dramatically. In contrast, for any step
$t$ during which there are at least $p$ level-$(i + 1)$-active cups,
the emptier's actions reduce the level-$i$ integer fill by exactly
$p$. Using this fact, we show that the only way for the filler to
achieve large level-$i$ integer fill $\ell$ after a step $t_1$ is
to perform at least $p(t_1 - t_0 + 1) + \Omega(\ell)$ level-$i$
threshold crossings during the steps $t_0, \ldots, t_1$ for some step $t_0
\le t_1$.

If the emptier crosses $p(t_1 - t_0 + 1) + s$ level-$i$ thresholds
during a sequence of steps $t_0, \ldots, t_1$, then we say that
$\max(s, 0)$ is the \defn{level-$i$ threshold bolus} during the step
interval. In order to prove Theorem \ref{thm:randomized_pre}, the
challenge becomes to bound the level-$i$ threshold bolus for all step
intervals. Specifically, if $m_i$ is the maximum number of level-$i$
active cups over all steps, then we wish to bound the maximum
level-$i$ threshold bolus of any step interval by at most
$O(m_i^{.75})$ (unless $m_i$ is already very small).

The initial offsets $r_j$ placed into each cup $j$ serve to randomize
when threshold crossings (at any level) occur in each cup
$j$. As a consequence, one can show that for any fixed set $S$ of
cups, and for any fixed step interval $t_0, \ldots, t_1$, the number
of threshold crossings (at any level) in cups $S$ during the
step interval is, with (very) high probability in $|S|$, at most $p(t_1 - t_0
+ 1) + O(|S|^{0.75})$.

A key technical difficulty, however, is that the random offsets $r_j$
partially determine which cups are active at each level. It may be,
for example, that the filler can design a strategy in order to ensure
that almost all of the level-$i$ active cups have abnormally small
random offsets $r_j$. Thus the set of level-$i$ active cups after a
given step cannot be analyzed as a fixed set $S$.

The gaps of size two between successive levels plays an important role
in resolving this issue. Consider a step interval $t_0, \ldots, t_1$,
and suppose that the filler achieves a large level-$i$ threshold bolus
during the step interval. For any cup $j$ that is not level-$i$ active
at step $t_0 - 1$, the cup requires at least $2$ units of water before
it can begin crossings level-$i$ thresholds. Such cups $j$ are, in
some sense, a poor investment for the filler, since the total number
of thresholds crossed in the cup $j$ during the step interval is
deterministically smaller than the amount of water placed into the
cup. In fact, if the filler crosses level-$i$ thresholds in more than
$A^{(i)}(t_0 - 1)$ different cups $j$ that are not initially active
after step $t_0 - 1$ (i.e., the filler makes more than $A^{(i)}(t_0 -
1)$ bad investments), then it becomes impossible for the filler to
even achieve a \emph{positive} level-$i$ threshold bolus during the
step interval. In order for the filler to achieve a large level-$i$
threshold bolus, it follows that almost all (i.e., all but $O(m_i)$
units) of the water placed by the filler during the step interval must
be in a set $S$ of at most $O(m_i)$ cups.

Now, suppose there exists a set $S$ of at most $O(m_i)$ cups into
which the filler places almost all (i.e., all but $O(m_i)$ units) of
their water during the step interval $t_0, \ldots, t_1$. Whereas the
set of level-$i$ active cups is dependent on the random offsets $r_j$,
the set $S$ is not. It follows that, if such a set $S$ exists for a
step interval $t_0, \ldots, t_1$, we can bound the number of threshold
crossings (at any level) by at most $p(t_1 - t_0 + 1) +
O(m_i^{0.75})$ with (very) high probability in $m_i$. (With a
bit of work, this includes the threshold crossings by the $O(m_i)$
units of water not in in $S$.) On the other hand, if no such set $S$
exists for the step interval, then the level-$i$ threshold bolus will
deterministically be zero.

The analysis described above can be formalized to show that, with
probability at least $1 - 2^{-\polylog(n)}$, every level $i$ for which
$m_i$ is sufficiently large in $\Omega(p \log n) + \polylog(n)$ has
the property that $m_{i + 1} \le O(m_i^{0.75})$. It follows that for
some level $i \in O(\log \log n)$, the maximum number of level-$i$
active cups will never exceed $O(p\log n) + \polylog(n)$. Using the
analysis of the \emph{deterministic} greedy algorithm, we get that
no cup $j$ will ever have level-$i$ fill greater than $O(\log p + \log
\log n)$.

\subsection{Lower Bound Constructions}
\label{sec:lowerbounds_overview}

The known lower-bound constructions for the $p$-processor cup game on
$n$ cups achieve backlog $\Omega(\log (n/p))$ \cite{BenderFaKu19}. In
the proof of Theorem \ref{thm:lognbacklog_pre} we demonstrate slightly
stronger bound of $\Omega(\log (n - p))$. Unlike the previous lower
bounds, this construction relies on the fact that the emptier has no
resource augmentation.

In order to achieve backlog $\Omega(\log (n - p))$, the filler uses
the following strategy: At each step they place $1$ unit of water in
each of cups $1,2,\ldots, p - 1$; the remaining unit of water is
distributed among the cups $p, p + 1, \ldots, n$ as though the emptier
were playing a single-processor cup game on just those cups. The
filler follows a strategy in the simulated single-processor game such
that, if the emptier always removes only one unit of water from the
cups $p, p + 1, \ldots, n$, then the filler will achieve backlog
$\Omega(\log (n - p))$ in those cups within $O(n - p)$ steps. If, on
the other hand, the emptier chooses to at some point remove multiple
units of water from the cups $p, p + 1, \ldots, n$ in a single step
$t$, then the emptier will be forced to neglect one of the cups $1, 2,
\ldots, p - 1$, meaning that $\tot_{[p - 1]}(S_t) \ge \tot_{[p -
    1]}(S_{t - 1}) + 1$; such a step $t$ is called a \defn{growth
  step}. Whenever a growth step occurs, the filler restarts their
single-processor strategy on the cups $p, p + 1, \ldots, n$ from
scratch. Eventually, there must either have been a large number of
growth steps, in which case one of the cups $1, 2, \ldots, p - 1$ will
have large fill, or there must be a long sequence of steps in which
there are no growth steps, in which case one of the cups $p, p + 1,
\ldots, n$ must achieve fill $\Omega(\log (n - p))$.

In Section \ref{sec:randomizedlowerbounds}, we generalize this
construction to prove lower bounds against randomized emptiers with an
oblivious filler. Lower bound constructions are already known that
achieve backlog $\Omega(\log \log \frac{n}{p})$ against any emptying
strategy with probability at least $\frac{1}{\poly(n)}$
\cite{BenderFaKu19}. The focus in Section
\ref{sec:randomizedlowerbounds} is to demonstrate two other ways in
which Theorem \ref{thm:randomized_pre} is tight, and in which the
$p$-processor cup game differs fundamentally from the $(\epsilon,
\delta)$-augmented $p$-processor cup game.

First, we design an oblivious strategy for the filler that achieves
backlog $\Omega(\log \log p)$ with \emph{constant} probability. This
contrasts with the $(\epsilon, \delta)$-augmented game, for which a an
algorithm is known to achieve backlog $O(1)$ with probability $1 -
e^{- \Omega(\epsilon^2 p)}$ \cite{BenderFaKu19} under certain natural
conditions on $\epsilon, \delta$.

To achieve backlog $\Omega(\log \log p)$, the filler's strategy is
roughly as follows. They maintain an \defn{anchor set} $A$ of $p - 1$
cups, and at each step they always place $1$ unit of water in each cup
in the anchor set. On the remaining cups, the filler performs a
single-processor randomized construction such that, if the emptier
never removes from more than one non-anchor cup per step, then there
will be steps in which the filler has a reasonable probability (say,
$\frac{1}{\sqrt{p}}$) of achieving backlog $\Omega(\log \log p)$ in
some non-anchor cup $j$. For each such step, the filler will, with
some small probability $\frac{1}{\poly(p)}$ modify the anchor set $A$
by swapping cup $j$ with a random current member of $A$. In order for
the emptier to avoid the anchor set $A$ having cups with large fill
swapped into it, the emptier must occasionally empty from multiple
non-anchor cups in the same step. However, if the emptier does this
too frequently, then this alone could increase the fills of the anchor
cups to be unacceptably large. No matter the emptier's strategy, there
is at least constant probability that some anchor-cup has fill
$\Omega(\log \log p)$ at the end of the filler's construction.

Our second lower bound against randomized emptying strategies
considers the possibility of a randomized emptying algorithm for the
multi-processor cup game providing an unending guarantee, meaning that
for any (arbitrarily large) time step $t \in \mathbb{N}$, the algorithm
gives an at-least-constant-probability bound of $o(\log n)$ on the
backlog at time $t$. Unending guarantees are known to be possible with
the help of resource augmentation \cite{BenderFaKu19}. We show that,
without resource augmentation, unending guarantees are impossible for
any ``greedy-like'' emptying strategy, including both the smoothed
greedy algorithm and the multi-processor algorithm given by
\cite{BenderFaKu19}. The construction is similar to the others
outlined above.

\section{Analyzing Greedy}\label{sec:greedy}

Theorem \ref{thm:lognbacklog} establishes that if the emptier follows
the greedy strategy, then the backlog in the $p$-processor cup game
never exceeds $O(\log n)$, and that this is tight for $n \ge
2p$.

\begin{theorem}[Theorem \ref{thm:lognbacklog_pre} Restated]
The greedy algorithm for the $p$-processor cup game on $n$ cups
achieves backlog $O(\log n)$. Moreover, there is an adaptive strategy
for the filler that achieves backlog $\Omega(\log (n - p))$ against
any emptying strategy.
\label{thm:lognbacklog}
\end{theorem}

As a convention, throughout the section, we define $I_t$ to be
the \defn{intermediate state} between steps $S_{t - 1}$ and $S_t$, in
which the filler has inserted water into the cups, but the emptier has
not yet removed water from the cups.

To prove Theorem \ref{thm:lognbacklog}, we begin by establishing a
weaker bound. Lemma \ref{lem:Mfinite}, which we prove in Subsection
\ref{sec:Mfinite}, bounds the fill of every cup by $O(p + \log n)$.
\begin{lemma}
  The greedy algorithm for the $p$-processor cup game on $n$ cups
  achieves backlog $O(p + \log n)$.
  \label{lem:Mfinite}
\end{lemma}
In addition to introducing technical ideas that reoccur later in the
section, Lemma \ref{lem:Mfinite} plays a small but important role in
the proof of Theorem \ref{thm:lognbacklog} by establishing that the
backlog achievable by the filler is \emph{finite}.

Critically, Lemma \ref{lem:Mfinite} enables us to define the quantity
$M \in \mathbb{R}^+$ to be,
\begin{equation}
M = \sup_{t \in \mathbb{N}, \, S_t \in \mathcal{S}_t} \av_p(S_t),
\label{eq:M}
\end{equation}
and to eliminate the possibility that $M = \infty$.

Lemma \ref{lem:Mtobacklog} establishes that, in order to bound the
backlog of each state $S_t$ by $O(\log n)$, it suffices to instead
bound $M$ by $O(\log n)$. Specifically, if a step $t$ achieves a
backlog $B$ that is greater than in the preceding step, then
Lemma \ref{lem:Mtobacklog} establishes that $\av_p(S_t) \ge B - 1$.

\begin{lemma}
If $S_t(1) > S_{t - 1}(1)$, then we have $S_t(1),
S_t(2), \ldots, S_t(p + 1) \ge S_t(1) - 1$.
\label{lem:Mtobacklog}
\end{lemma}
\begin{proof}
Let $A$ denote the $p$ fullest cups in $I_t$. The emptier will remove
$1$ unit of water from each cup in $A$, resulting in the cup
containing no more water than it contained in state $S_{t -
1}$. Therefore, in order for the fullest cup $i$ in $S_t$ to be fuller
than the fullest cup in $S_{t - 1}$, it must be that $i \notin A$. Since each cup in $A$ has fill greater than the fill $B$ of cup
$i$ in state $I_t$, the cups in $A$ must continue to have fill at
least $B - 1$ in state $S_t$, as desired.
\end{proof}

For $1 \le k \le n - p$, and $N \in \mathbb{R}^+$, define the
\defn{$N$-skewed average} of the $p + k$ fullest cups in a state $S$
to be,
\begin{equation}
f^N_k(S) = \max\left(\frac{\tot_{p + k}(S) - p \cdot N}{k}, 0\right).
\label{eq:fk}
\end{equation}
When $N = M$, we will omit $N$, instead using $f_k(S)$ to denote the
$M$-skewed average.

Combinatorially, one should think of $f_k(S)$ as follows: Take the
total amount of water in the $p + k$ fullest cups; move as much of
that water to the $p$ fullest cups as possible, without allowing the
average fill of the $p$ fullest cups to exceed $M$ \footnote{If the
  average fill of the $p$ fullest cups initially exceeds $M$, then
  this means we are actually moving water from the $p$-fullest cups to
  the cups with ranks $p + 1, \ldots, p + k$.}; then distribute the
remaining water evenly among the cups of rank $p + 1, \ldots, p +
k$. The amount of water in each of the cups of rank $p + 1, \ldots, p
+ k$ is the skewed average $f_k(S)$.

Lemma \ref{lem:boundMbyskewaverage} exploits a (non-obvious)
combinatorial relationship between $M$ and $f_1(S_t)$ in order to
bound the former in terms of the latter.

\begin{lemma}
Suppose that $f_1(S_t) \le K$ for all $t \in \mathbb{N}$ and
$S_t \in \mathcal{S}_t$. Then,
$$M \le K + O(\log p).$$
\label{lem:boundMbyskewaverage}
\end{lemma}

Combined, Lemmas \ref{lem:Mtobacklog} and
\ref{lem:boundMbyskewaverage} reduce the proof of the upper bound on
backlog in Theorem \ref{thm:lognbacklog} to proving that $f_1(S_t) \le
O(\log n)$ for all $S_t \in \mathcal{S}_t$. This is accomplished by Lemma
\ref{lem:skewaverage}, which bounds each of $f_1(S_t), f_2(S_t),
\ldots, f_{n - p}(S_t)$ for each step $t$.

\begin{lemma}
For all $t \in \mathbb{N}$, for all $S_t \in \mathcal{S}_t$, and for
all $1 \le k \le n - p$,
$$f_k(S_t) \le 1 + \left(\frac{1}{k + 1} + \frac{1}{k + 2} + \cdots + \frac{1}{n}\right).$$
\label{lem:skewaverage}
\end{lemma}


The remainder of the section proceeds as follows. Subsection
\ref{sec:Mfinite} proves Lemma \ref{lem:Mfinite}, establishing an
upper bound on backlog of $O(p + \log n)$, and thus demonstrating that
$M$ is finite. Subsections \ref{sec:lemboundM} and
\ref{sec:lemskewaverage} prove Lemmas \ref{lem:boundMbyskewaverage}
and \ref{lem:skewaverage}, which together bound the backlog by $O(\log
n)$. Subsection \ref{sec:lowerbound} then completes the proof of
Theorem \ref{thm:lognbacklog} by constructing an adaptive strategy for
the filler that achieves backlog $\Omega(\log (n - p))$.

In Appendix \ref{sec:appgreedyintuition} we discuss the structural
relationship between our analysis of the $p$-processor game for $p >
1$, and the previously known analysis of the single-processor
game. The discussion will likely be most useful to the reader after
they have read Subsection \ref{sec:Mfinite}.

\subsection{Proof of Lemma \ref{lem:Mfinite}}\label{sec:Mfinite}

In this section, we prove that the greedy emptying algorithm for the
$p$-processor cup game on $n$ cups achieves a backlog of $O(p + \log
n)$.





Define the \defn{$N$-truncated} $p$-processor cup game to be the
standard $p$-processor cup game, except with one additional
restriction on the filler: the filler is never permitted to increase
the fill of any cup to be larger than $N$. The next lemma bounds the
$N$-skewed averages $f^N_k(S_t)$ after each step $t$ in the
$N$-truncated cup game.

\begin{lemma}
Suppose the emptier follows the greedy algorithm for the the
$N$-truncated $p$-processor cup game on $n$ cups. Then, for all $t \ge 0$, for all $S_t \in \mathcal{S}_t$, and for all $k \in [n - p]$,
$$f^N_k(S_t) \le 1 + \left(\frac{1}{k + 1} + \frac{1}{k + 2} + \cdots + \frac{1}{n}\right).$$
\label{lem:truncatedinvariant}
\end{lemma}

Before proving Lemma \ref{lem:truncatedinvariant}, we describe how to
use it in order to bound the backlog in the (non-truncated)
$p$-processor game. Suppose that during an $N$-truncated $p$-processor
game on $n$ cups, there is some step $t$ such that $S_t(1) \ge N - 2$,
and let $t$ be the first such step. (Importantly, the truncation
restriction on the filler has not actually affected the filler by step
$t$.) Then by Lemma \ref{lem:Mtobacklog}, all of the cups with ranks
$1, 2, \ldots, p + 1$ in state $S_t$ must contain fill at least $N -
3$. The skewed average $f^N_1(S_t)$ is therefore
\begin{align*}
  & \max\left(\frac{\tot_{p + 1}(S_t) - p \cdot N}{1}, 0 \right) \\ & \ge (p
  + 1)(N - 3) - p \cdot N \ge N - 3 - 3p.
\end{align*}
By Lemma \ref{lem:truncatedinvariant}, however, this means that $N$
can be at most $O(\log n + p)$. Thus the only values $N$ for which the
filler can at some point achieve backlog at least $N - 3$ in the
$N$-truncated game satisfy $N \le O(p + \log n)$, which completes the
proof of Lemma \ref{lem:Mfinite}.

We complete the subsection by proving Lemma
\ref{lem:truncatedinvariant}.
\begin{proof}[Proof of Lemma \ref{lem:truncatedinvariant}]
We prove the result by induction on $t$. As a base case, when $t = 0$,
the lemma holds trivially, since $f^N_k(S_0) = 0$ for all
$k$. Consider $t > 0$, and suppose as an inductive hypothesis that the
result holds for $t - 1$.

Suppose that $I_t(p) \le 1$, meaning that at least one of the $p$
fullest cups in state $I_t$ contains fill $1$ or smaller. Then, in
state $S_t$, the $p$ fullest cups each contain fill at most $N$ (by
$N$-truncation) and the remaining cups each contain fill at most
$1$. This implies that $f^N_k(S_t) \le 1$ for all $k > 0$. It follows
that, for the remainder of the proof, we may assume that $I_t(1),
\ldots, I_t(p) > 1$.

To get from state $I_t$ to state $S_t$, one unit of water is removed
from each of the $p$ fullest cups. Of these $p$ cups, let $A$ denote
the subset that remain among the $p + k$ fullest cups in $S_t$,
and let $B$ denote the subset that do not. Let $C$ denote the cups in
$I_t$ with ranks $p + 1, \ldots, p + k$. Notice that $A \cup B \cup C$
comprise the $p + k$ fullest cups in state $I_t$.

Since $I_t$ differs from $S_{t - 1}$ by the insertion of $p$ units of water,
\begin{equation}
  f^N_{k + |B|}(I_t) \le f^N_{k + |B|}(S_{t - 1}) + \frac{p}{k + |B|}.
  \label{eq:primefIt}
\end{equation}

Let $I'_t$ denote an intermediate state between $I_t$ and $S_t$ at
which one unit of water has been removed from each cup in $A$, but not
from each cup in $B$. By the definition of $A$, the $p + k$ fullest
cups in $I'_t$ are the same as those in $I_t$. Since $I'_t$ contains
$|A|$ fewer units of water in cups $A$ than does $I_t$, either $f^N_{k +
  |B|}(I'_t) = 0$, or
\begin{equation}
  f^N_{k + |B|}(I'_t) \le f^N_{k + |B|}(I_t) - \frac{|A|}{k + |B|}.
  \label{eq:primefIprimet}
\end{equation}
Combining \eqref{eq:primefIt} and \eqref{eq:primefIprimet} yields,
\begin{align*}
 f^N_{k + |B|}(I'_t) & \le f^N_{k + |B|}(S_{t - 1}) + \frac{p - |A|}{k +
    |B|}\\ & = f^N_{k + |B|}(S_{t - 1}) + \frac{|B|}{k + |B|}. 
\end{align*}

To complete the proof, we wish to show that
\begin{equation}
  f^N_k(S_t) \le f^N_{k + |B|}(I'_t).
  \label{eq:primefkplusBtojustk}
\end{equation}

In particular, by the inductive hypothesis for $t - 1$, this would yield
\begin{equation*}
  \begin{split}
    f^N_k(S_t) & \le 1 + \left(\frac{1}{k + |B| + 1} + \frac{1}{k + |B| + 2} + \cdots + \frac{1}{n - p}\right) \\ & \phantom{foob} + \frac{|B|}{k + |B|} \\
    & \le 1 + \left(\frac{1}{k + 1} + \frac{1}{k + 2} + \cdots + \frac{1}{n - p}\right),
  \end{split}
\end{equation*}
as desired.

In proving \eqref{eq:primefkplusBtojustk}, it suffices to prove that,
in state $I'_t$, the average fill of the cups in $B$ is at least
$f^N_{k + |B|}(I'_t)$. That is,
\begin{equation}
  \av_B(I'_t) \ge f_{k + |B|}(I'_t).
  \label{eq:avBoriginal}
\end{equation}
In particular, \eqref{eq:avBoriginal} means that
$$\tot_{p + k}(S_t) \le \tot_{p + k + |B|}(I'_t) - |B| \cdot f_{k + |B|}(I'_t),$$
and it would follow that, as long as $f_k(S_t) > 0$, then
\begin{equation*}
  \begin{split}
    k \cdot f_k(S_t) & = \tot_{p + k}(S_t) - p \cdot M  \\
    & \le \left(\tot_{p + k + |B|}(I'_t) - p \cdot M\right) - |B| \cdot f_{k + |B|}(I'_t) \\
    & \le (k + |B|) \cdot f_{k + |B|}(I'_t) - |B| \cdot f_{k + |B|}(I'_t) \\
    & = k \cdot f_{k + |B|}(I'_t),
  \end{split}
\end{equation*}
which implies \eqref{eq:primefkplusBtojustk}.

We complete the proof by establishing \eqref{eq:avBoriginal}. The key
is to exploit the $N$-truncation restriction on the filler, which
guarantees that $\tot_p(I'_t) \le N \cdot p$. It follows that, unless
$f_{k + |B|}^N(I'_t) = 0$,
\begin{equation*}
  \begin{split}
    f^N_{k + |B|}(I'_t) & = \frac{\tot_{p + k}(I'_t) - p \cdot N}{k + |B|} \\
    & \le \frac{\tot_{p + k}(I'_t) - \tot_p(I'_t)}{k + |B|} \\
    & = \frac{S_{p + 1}(I'_t) + \cdots + S_{p + k + |B|}(I'_t)}{k + |B|}.
  \end{split}
\end{equation*}
The final expression is the average fill of the cups with ranks
$p + 1, \ldots, p + k + |B|$ in $I'_t$. Since the cups $B$ each
contain more fill in state $I'_t$ than do any of the cups with ranks $p + 1,
\ldots, p + k + |B|$, the cups $B$ must each contain fill at least
$f^N_{k + |B|}(I'_t)$. This completes the proof of \eqref{eq:avBoriginal}.
\end{proof}

\subsection{Proof of Lemma \ref{lem:boundMbyskewaverage}}\label{sec:lemboundM}

In this section, we prove Lemma \ref{lem:boundMbyskewaverage}, which
establishes that for some step-number $t$ and some state $S_t \in
\mathcal{S}_t$, $f_1(S_t) \ge M - O(\log p)$. This means that the
quantity $M$ can be bounded indirectly by proving a bound on
$f_1(S_t)$.

Call a step $t$ a \defn{record setter} if $\av_p(S_t)$ is larger than any
of $\av_p(S_1), \ldots, \av_p(S_{t - 1})$. The key to proving Lemma
\ref{lem:boundMbyskewaverage} is to show that, whenever a step $t$ is
a record setter, all of the fills $S_t(1), \ldots, S_t(p + 1)$ must be
quite close to each other. Specifically, we show that $S_t(1) - S_t(p
+ 1) \le O(\log p)$. This implies that when a record-setting step $t$
achieves $\av_p(S_t) \ge M - \epsilon$, for $\epsilon$ sufficiently
small, then
$$f_1(S_t) \ge S_t(p + 1) - p \cdot \epsilon \ge M - O(\log p),$$ as desired.

In order to bound $S_t(1) - S_t(p + 1)$ for a record-setting step $t$,
we begin by proving a generalization of Lemma \ref{lem:Mtobacklog}.
\begin{lemma}
Suppose that $S_t(j) > S_{t - 1}(j)$ for some $j \le p$. Then $S_t(j +
1), S_t(j + 2), \ldots, S_t(p + 1)$ are each at least $S_t(j) - 1$.
\label{lem:cupreset}
\end{lemma}
\begin{proof}
Let $R$ denote the $p$ fullest cups in state $I_t$. Between states
$I_t$ and $S_t$, the emptier removes $1$ unit of water from each cup
in $R$, ensuring that each cup $i \in R$ contains at most as much
water after step $t$ as after step $t - 1$. In order so that $S_t(j) >
S_{t - 1}(j)$, it follows that the cups of rank $1, 2, \ldots, j$ in
$S_t$ cannot all be in $R$ (otherwise, the $j$ fullest cups in $S_{t}$
would each contain at least as much water in $S_{t - 1}$ as they do in
$S_t$). This means that the rank-$(p + 1)$ cup in $I_t$ is among the $j$
fullest cups in $S_t$. Since each of the $p + 1$ fullest cups in $I_t$
contains at least $S_t(j)$ water, the same $p + 1$ cups will continue
to have fill at least $S_t(j) - 1$ in state $S_t$. Thus each of $S_t(j
+ 1), S_t(j + 2), \ldots, S_t(p + 1)$ are at least $S_t(j) - 1$.
\end{proof}

Using Lemma \ref{lem:cupreset} as a building block, we next prove a
set of constraints on the fills $S_t(1), \ldots, S_t(p + 1)$ of any
record-setting step $t$.
\begin{lemma}
  For any record-setting step $t$, and for any $i \in [p]$,
  $$\frac{S_t(i + 1) + \cdots + S_t(p + 1)}{p + 1 - i} \ge S_t(i) - 1.$$
  \label{lem:recordsettingconstraints}
\end{lemma}
\begin{proof}
  For $i \le p$, say that rank $i$ is \defn{reset} during a step $k$
  if $S_k(i) > S_{k - 1}(i)$. By Lemma \ref{lem:cupreset}, whenever a
  rank $i$ is reset during a step $k$, the average amount of water in
  the cups with ranks $i + 1, i + 2, \ldots, p + 1$ must be at least
  $S_k(i) - 1$.

  For $i = 1, \ldots, p$, define $t_i$ to be the first step such that
  the $i$ fullest cups in $S_{t_i}$ each contain at least as much
  water as the $i$ fullest cups in $S_t$; that is,
  $$S_{t_i}(1) \ge S_t(1), \, S_{t_i}(2) \ge S_t(2), \ldots, \,
  S_{t_i}(i) \ge S_t(i).$$

  By the definition of $t_i$, at least one rank $j \in [i]$ must be
  reset at step $i$. By Lemma \ref{lem:cupreset}, it follows that the
  average amount of water in each of the cups with ranks $i + 1, i +
  2, \ldots, p + 1$ in $S_{t_i}$ is at least $S_{t_i}(i) - 1 \ge
  S_t(i) - 1$. Since state $S_t$ is a record setter, it follows that
  the average fill of the cups with ranks $i + 1, i + 2, \ldots, p +
  1$ in state $S_t$ is also at least $S_t(i) - 1$.
\end{proof}

Using the constraints from Lemma \ref{lem:recordsettingconstraints},
we can now bound $S_t(1) - S_t(p + 1)$ for a record-setting step $t$.
\begin{lemma}
For any record-setting step $t$, $S_t(1) - S_t(p + 1) \le O(\log p)$.
 \label{lem:boundpplus1cup}
\end{lemma}
\begin{proof}
  For $j = 1, \ldots, p$, define $\Delta_j = S_t(p + 1 - j) - S_t(p +
  2 - j)$ to be the difference in fill between the rank-$(p + 1 - j)$
  and rank-$(p + 2 - j)$ cups in state $S_t$. Lemma
  \ref{lem:recordsettingconstraints} establishes that for each $i = 1,
  \ldots, p$, the average fill of the cups with ranks $p + 2 - i,
  \ldots, p + 1$ is at most one smaller than the fill of the rank-$(p
  + 1 - i)$ cup in $S_t$. This enforces the inequality,
  \begin{equation}
    \sum_{j = 1}^{i} \Delta_j \cdot j \le i,
    \label{eq:deltaineq}
  \end{equation}
  for all $i = 1, \ldots, p$. In order to complete the proof, we will
  show that the inequalities given by \eqref{eq:deltaineq} together
  imply that
  \begin{equation}
    \sum_{j = 1}^p \Delta_j \le O(\log p).
    \label{eq:deltasum}
  \end{equation}

  Intuitively, under the constraints of \eqref{eq:deltaineq}, the
  objective function $\sum_j \Delta_j$ is maximized by setting each
  $\Delta_j = \frac{1}{j}$. To see this formally, consider values
  $\Delta_1, \ldots, \Delta_p \in [0, 1]$ that satisfy
  \eqref{eq:deltaineq}. Suppose that $\Delta_i \neq \frac{1}{i}$ for
  some $i$. We will show that there are values $\Delta_1', \ldots,
  \Delta_p'$ satisfying \eqref{eq:deltaineq} such that $\sum_i
  \Delta'_i \ge \sum_i \Delta_i$ and such that the number of $i$ for
  which $\Delta_i' = \frac{1}{i}$ is greater than the number of $i$
  for which $\Delta_i = \frac{1}{i}$. Iteratively applying this
  observation, it follows that the objective function $\sum_i
  \Delta_i$ is maximized by setting $\Delta_i = \frac{1}{i}$ for each
  $i$.

  Suppose that $\Delta_i \neq \frac{1}{i}$ for some $i$, and let $i$
  be the smallest $i$ for which $\Delta_i > \frac{1}{i}$. (Note that
  if $\Delta_i \le \frac{1}{i}$ for all $i$ then we can simply define
  $\Delta'_i = \frac{1}{i}$ for all $i$ to complete the proof.) Then
  by \eqref{eq:deltaineq}, there must be some $j < i$ for which
  $\Delta_j < \frac{1}{j}$. Set $\Delta_1', \ldots, \Delta_p'$ so that
  $\Delta'_k = \Delta_k$ for $k \neq i, j$, and such that
  $$\Delta'_i = \Delta_i - \frac{\min(i \cdot (\Delta_i - 1/i),
    j \cdot (1/j - \Delta_j))}{i},$$
  and
  $$\Delta'_j = \Delta_i + \frac{\min(i \cdot (\Delta_i - 1/i), j
    \cdot (1/j - \Delta_j))}{j}.$$ Notice that at least one of
  $\Delta'_i$ or $\Delta'_j$ is equal to $\frac{1}{i}$ or
  $\frac{1}{j}$, respectively, as desired. Moreover, since $j < i$,
  the objective function $\sum_i \Delta'_i$ is larger than the
  objective function $\sum_i \Delta_i$. To complete the proof, it
  remains to establish that the $\Delta'_k$'s satisfy
  \eqref{eq:deltaineq}.

  Notice that $\Delta_k' \le \frac{1}{k}$ for $k < i$, ensuring that
  $$\sum_{k = 1}^{r} \Delta'_k \cdot k \le r,$$
  for each $r < i$. On the other hand, for $r \ge i$,
  $$\sum_{k = 1}^{r} \Delta'_k \cdot k = \sum_{k = 1}^{r} \Delta_k
  \cdot k,$$ since $i \cdot \Delta'_i + j \cdot \Delta'_j = i \cdot
  \Delta_i + j \cdot \Delta_j$. Therefore, the $\Delta'_k$'s satisfy
  the inequalities \eqref{eq:deltaineq}, as desired.
\end{proof}

We now complete the proof of Lemma \ref{lem:boundMbyskewaverage}.
\begin{proof}[Proof of Lemma \ref{lem:boundMbyskewaverage}]
We wish to construct a state $S_t \in \mathcal{S}_t$ for some $t$ such
that
$$M \le f_1(S) + O(\log p).$$

Let $\epsilon > 0$. By the definition of $M$, there must exist a
series of steps with states $S_1, \ldots, S_t$ such that $\av_p(S_1),
\ldots, \av_p(S_{t - 1}) < M - \epsilon$, but $\av_p(S_t) \ge M -
\epsilon$. By Lemma \ref{lem:boundpplus1cup}, since step $t$ is a
record setter and since $S_t(1) \ge M - \epsilon$, we have
\begin{equation}
  S_t(p + 1) \ge M - \epsilon - O(\log p).
  \label{eq:pplus1cupfull}
\end{equation}

On the other hand, since the average of $S_t(1), \ldots, S_t(p)$ is at
least $M - \epsilon$, the amount of water that oen can move from the
cup with rank $p + 1$ in state $S_t$ to the cups with ranks $1,
\ldots, p$, without increasing $\av_p$ to more than $M$ is at most
$\epsilon \cdot p$. It follows that
\begin{equation}
  f_1(S_t) \ge S_t(p + 1) - \epsilon \cdot p.
  \label{eq:firstcupfull}
\end{equation}

Setting $\epsilon$ to be sufficiently small, \eqref{eq:pplus1cupfull}
and \eqref{eq:firstcupfull} imply that
$$f_1(S_t) \ge M - O(\log p),$$
as desired.
\end{proof}

\subsection{Proof of Lemma \ref{lem:skewaverage}}\label{sec:lemskewaverage}

In this section, we wish to prove:
\begin{lemma}[Lemma \ref{lem:skewaverage} restated]
For all $t \in \mathbb{N}$, for all $S_t \in \mathcal{S}_t$, and for
all $1 \le k \le n - p$,
$$f_k(S_t) \le 1 + \left(\frac{1}{k + 1} + \frac{1}{k + 2} + \cdots + \frac{1}{n}\right).$$
\end{lemma}

We follow the same inductive strategy as in the proof of Lemma
\ref{lem:Mfinite}. 

As in the proof of Lemma \ref{lem:Mfinite}, we consider the case in
which $I_t(p) \le 1$ separately. By the definition of $M$, we know
that $\av_p(S_t) \le M$. If $I_t(p) \le 1$, then it follows that
$S_t(p), S_t(p + 1), \ldots, S_t(n) \le 1$. Combining this with the
fact that $\av_p(S_t) \le M$, it follows that $f_k(S_t) \le 1$ for all
$k > 0$.

Now consider the case of $I_t(p) > 1$. Define the sets $A, B, C$ and
the state $I'_t$ as in the proof of Lemma \ref{lem:Mfinite}. In order
to prove the lemma, it suffices to demonstrate that,
\begin{equation}
  \av_B(I'_t) \ge f_{k + |B|}(I'_t),
  \label{eq:avB}
\end{equation}
which can then be used to complete the proof just as in Lemma
\ref{lem:Mfinite}.

Unlike in the proof of Lemma \ref{lem:Mfinite}, however, we cannot
assume that, in state $I'_t$, each of the cups in $A \cup B$ contains
fill at most $M$. (Recall that Lemma \ref{lem:Mfinite} considered only
the $N$-truncated game, which forced each cup's fill to, by
definition, never exceed $N$.) In general, a cup in $B$ at state
$I'_t$ could contain fill up to one greater than the same cup in state
$S_{t - 1}$, allowing for $\av_B(I'_t)$ to be as large as $M + 1$. To
see why having $\av_B(I'_t) > M$ may be a problem, suppose that in
state $I'_t$, the cups in $A$ each contain $M$ units of water, and the
cups in $B \cup C$ each contain $M + \epsilon$ units of water for some
$\epsilon > 0$. Then $f_{k + |B|}(I'_t) = M + \epsilon +
\frac{\epsilon |B|}{|C|} > \av_B(I'_t)$. On the other hand, in this
example, the average fill of the $p$ fullest cups in state $S_t$
satisfies $\av_p(S_t) > M$, which contradicts the definition of
$M$. This suggests that the key to proving \eqref{eq:avB} is to
exploit the requirement $\av_p(S_t) \le M$.

Define $b$ so that, in state $I'_t$, if each cup in $B$ had fill $b$,
then $\av_{A \cup B}(I'_t)$ would be $M$. That is,
$$\av_{A}(I'_t) \cdot |A| + b \cdot |B| = M.$$ By the requirement that
$\av_p(S_{t + 1}) \le M$, it must be that, in state $I'_t$, the
average fill of the $|B|$ fullest cups in $C$ is at most $b$. This, in
turn, implies that
\begin{equation}
  \av_C(I'_t) \le b.
  \label{eq:b}
\end{equation}
Define $\delta \in \mathbb{R}$ so that
$\av_B(I'_t) = b + \delta$. In state $I'_t$, the total amount of water
in cups $A \cup B$ can be written as $M \cdot p + \delta \cdot
|B|$. Thus
$$f_{k + |B|}(I'_t) \le \frac{\delta \cdot |B| + \tot_C(I'_t)}{k + |B|} = \frac{\delta \cdot |B|}{k + |B|} + \av_C(I'_t).$$
If $\delta \le 0$, then it follows that $f_{k + |B|}(I'_t) \le \av_C(I'_t) \le \av_B(I'_t)$, which implies \eqref{eq:avB} as desired. If, on the other hand, $\delta > 0$, then since \eqref{eq:b} gives $\av_C(I'_t) \le b$, we have
$$f_{k + |B|}(I'_t) \le \frac{\delta \cdot |B|}{k + |B|} + \av_C(I'_t) \le \delta + b = \av_B(I'_t),$$
which again implies \eqref{eq:avB}, as desired.

\subsection{A lower bound of $\Omega(\log (n - p))$}\label{sec:lowerbound}

The known lower-bound constructions for the $p$-processor cup game on
$n$ cups achieve backlog $\Omega(\log (n/p))$ \cite{BenderFaKu19}. In this section, we
prove a slightly stronger bound of $\Omega(\log (n - p))$, which
completes the proof of Theorem \ref{thm:lognbacklog}, establishing that
the greedy algorithm for the empty is asymptotically optimal when $n
\ge 2p$.

Formally, we prove the following proposition:
\begin{proposition}
  In the $p$-processor cup game on $n$ cups, there exists an adaptive
  strategy for the filler that guarantees backlog at least
  $$\left(\frac{1}{2} + \frac{1}{3} + \cdots + \frac{1}{n - p + 1}\right) \ge \Omega(\log (n - p))$$
  after some step, regardless of the strategy followed by the emptier.
  \label{prop:lowerbound}
\end{proposition}

We begin with a well-known construction in the case of $p = 1$
\cite{BenderFaKu19, BenderFeKr15, DietzRa91}, which we reproduce for completeness.
\begin{lemma} 
  In the $1$-processor cup game on $n$ cups, starting in any initial
  cup state $S_0$, there is a strategy that the filler can follow so
  that after step $n - 1$, some cup has fill at least
  $$\frac{1}{2} + \frac{1}{3} + \frac{1}{4} + \cdots + \frac{1}{n}.$$
  \label{lem:lowerboundlemma}
\end{lemma}
\begin{proof}
  In the first step, the filler places $1/n$ water in each of the $n$
  cups. In the second step, the filler places $1/(n - 1)$ water in
  each of the $n - 1$ cups that have not yet been emptied
  from. Continuing like this, in the $i$-th step, the filler places $1
  / (n - i + 1)$ units of water in each of the cups that have not yet
  been emptied from. After the $(n - 1)$-th step, there is one cup
  containing fill at least
  $$\frac{1}{2} + \frac{1}{3} + \frac{1}{4} + \cdots + \frac{1}{n}$$
  greater than it contained in $S_0$.
\end{proof}

Using Lemma \ref{lem:lowerboundlemma} as a subprocedure, we construct
our strategy for the filler in the $p$-processor game.
\begin{proof}[Proof of Proposition \ref{prop:lowerbound}]
Let $A = \{1, 2, \ldots, p - 1\}$ and $B = \{p, p + 1, p + 2, \ldots,
n\}$. At each step, the filler places one unit in each cup from $A$,
and then distributes up to one unit (using a yet to be described
strategy) to the cups in $B$.

Between successive steps, the amount of water in the cups $A$ can
never decrease. Call a step $t$ a \defn{growth step} if the emptier
removes water from more than one cup in $B$. Since at least one cup in
$A$ must be neglected by the emptier during a growth step, the total
fill of the cups in $A$ grows by at least $1$.

The filler's strategy proceeds in phases, with a new phase beginning
after each growth step. In each phase, the filler applies Lemma
\ref{lem:lowerboundlemma} to the $n - p + 1$ cups $B$, until either
there is a step in which the emptier removes a unit of water from more
than one cup in $B$, or until $n - p$ steps have passed. In the
former case, a growth step occurs, and in the latter case, some
cup in $B$ has fill at least,
\begin{equation}
  \frac{1}{2} + \frac{1}{3} + \cdots + \frac{1}{n - p + 1}.
  \label{eq:backlogend}
\end{equation}

It follows that, until a backlog of \eqref{eq:backlogend} is achieved,
the filler can continue to generate additional growth steps. After
sufficiently many growth steps, a backlog of \eqref{eq:backlogend}
must be achieved simply by virtue of the water in the cups $A$,
completing the proof.
\end{proof}

\section{Analyzing Smoothed Greedy}
\label{sec:smoothedgreedy}

In this section, we analyze the \defn{smoothed greedy algorithm} for
the $p$-processor cup game on $n$ cups \cite{BenderFaKu19}. At the beginning of
algorithm, the emptier selects independent values $r_j$ uniformly at
random from the interval $[0, 1]$ for each $j = 1, \ldots, n$. Prior
to the first step of the game, the emptier inserts $r_j$ water into
each cup $j$. The emptier's strategy at the end of each step is then
to simply remove $1$ unit of water from each of the $p$ fullest
cups. If, however, one or more of the $p$ fullest cups contains less
than $1$ unit of water, then the emptier does not remove from those
cups. This ensures the important property that the fractional amount
of water (i.e., the amount of water modulo $1$) in each cup $j$ after
a step $t$, is a function only of the initial random offset $r_j$ and
the filler's actions in the first $t$ steps.

It was shown by Bender et al. \cite{BenderFaKu19} that, if the emptier
has $\epsilon \ge \frac{1}{\polylog n}$ resource augmentation, then
with high probability in $n$ the smoothed greedy algorithm achieves
backlog $O(\log \log n)$ in each step of the single-processor cup
game.

In this section we present a new analysis of the smoothed greedy
algorithm that applies without the use of resource augmentation, and
to $p > 1$. Prior to this result, no sub-logarithmic bounds were known
for any randomized algorithm without resource augmentation, including in
the case of $p = 1$.

\begin{theorem}[Theorem \ref{thm:randomized_pre} Restated]
  Let $c$ be at least a sufficiently large constant. Then with
  probability at least $1 - \exp\left(\log^c n\right)$, the smoothed
  greedy algorithm achieves backlog $O(\log p + c \log \log n)$ after
  all of the first $\exp(\log^c n)$ steps.
  
  \label{thm:randomized}
\end{theorem}

It is known that an oblivious filler can achieve backlog $\Omega(\log
\log \frac{n}{p})$ (against any emptying strategy) with probability at
least $\frac{1}{\poly(n)}$ \cite{BenderFaKu19}. Theorem \ref{thm:randomized}
matches this bound when $n$ is large relative to $p$ (i.e., $p \le
\polylog n$). Closing the gap between the upper and lower bounds when
$p \gg \polylog(n)$ remains an open question.

The proof of Theorem \ref{thm:randomized} analyzes $\Theta(\log \log
n)$ cup games concurrently, where the participant cups in the
level-$i$ cup game, also known as the \defn{level-$i$ active cups},
are the cups containing fill $2(i - 1)$ or greater. We denote the
\defn{number of level-$i$ active cups after step $t$ by
  $A^{(i)}(t)$}. If $m_i$ is the maximum number of level-$i$ active
cups during any of the first $2^{\polylog(n)}$ steps, then we wish to
show that $m_{i + 1} \le m_i^{0.75}$ (unless $m_i$ is already quite
small).

Define $h_j(t)$ to be the height (i.e. fill) of cup $j$ after step
$t$. Define the \defn{level-$i$ fill} $h^{(i)}_j(t)$ of cup $j$ after
step $t$ to be $\max(h_j(t) - 2 (i - 1), 0)$.  The set of
level-$(i + 1)$ active cups after step $t$ is precisely the set of
cups for which $h^{(i)}_j(t) \ge 2$.

Rather than bounding the number of level-$(i + 1)$ active cups
directly, we instead bound a larger quantity that we call the
level-$i$ integer fill. The \defn{level-$i$ integer fill} after step
$t$, denoted by $T^{(i)}(t)$, is given by
$$T^{(i)}(t) = \sum_{j = 1}^n \max\left(\lfloor h^{(i)}_j(t) - 1
\rfloor, 0 \right).$$ We say a cup \defn{crosses a level-$i$
  threshold} $s$ at step $t$, whenever the amount of water $f$ placed
by the filler into cup $j$ during step $t$ satisfies $h^{(i)}_j(t - 1)
< s \le h^{(i)}_j(t - 1) + f$ for some integer $s \ge 2$. One can
think of the level-$i$ integer fill as counting the number of
level-$i$ threshold crossings that have not yet been undone by the
emptier.

If the emptier crosses $p(t_1 - t_0 + 1) + s$ level-$i$ thresholds
during a sequence of steps $t_0, \ldots, t_1$, then we say that
$\max(s, 0)$ is the \defn{level-$i$ threshold bolus} during the step
interval. We begin by showing that, in order for the filler to achieve
a large level-$i$ integer fill $T^{(i)}(t_1)$ at some step $t_1$,
there must be a step interval $t_0, \ldots, t_1$ in which the filler
achieves a level-$i$ threshold bolus of size close to
$T^{(i)}(t_1)$. When $p = 1$, we can simply select $t_0$ to be the
largest $t_0 < t_1$ for which $T^{(i)}(t_0 - 1) = 0$. This ensures
that during each step $t_0, \ldots, t_1$, the emptier makes progress
at least $1$ at decreasing the level-$i$ integer fill, which therefore
forces the filler to cross $(t_1 - t_0 + 1) + T^{(i)}(t_1)$
level-$i$ threshold crossings in order to achieve level-$i$ integer
fill $T^{(i)}(t_1)$ after step $t_1$. The same choice of $t_0$ does
not work for $p > 1$, however, since there may be steps in the
interval $t_0, \ldots, t_1$ during which fewer than $p$ cups are
level-$i$ active, and thus the emptier makes progress less than $p$ at
decreasing the integer fill. Lemma \ref{lem:fillerprogress} uses an
alternative choice for $t_0$ to prove the result for $p > 1$.

\begin{lemma}
Let $d$ be a sufficiently large constant. For any step $t_1$, there
must be some $t_0 \le t_1$ such that the number of level-$i$
thresholds crossed in steps $t_0, \ldots, t_1$ is at least
$$p(t_1 - t_0 + 1) + T^{(i)}(t_1) - d p \log n.$$
\label{lem:fillerprogress}
\end{lemma}
\begin{proof}
By Theorem \ref{thm:lognbacklog}, there exists some constant $d$ such that no
cup $j$ ever has fill greater than $d \log n$. (Note that Theorem
\ref{thm:lognbacklog} requires that the cups initially be empty, which
is not the case for the smoothed greedy algorithm; nonetheless,
Theorem \ref{thm:lognbacklog} can instead be applied to the water
sitting above height one in each cup, treating cups that contain $1$
or less units of water as empty.)

It follows that if $T^{(i)}(t) > d(p - 1) \log n$ for some step $t$,
then after step $t$ there must be at least $p$ cups that are level-$i$
active. Thus during step $t$, the emptier is able to remove water from
$p$ cups, decreasing the level-$i$ integer fill by $p$.
  
Let $t_0$ be the largest $t_0 \le t_1$ such that $T^{(i)}(t_0 - 1) \le
d (p - 1) \log n$. Then during each step $t_0, \ldots, t_1$, the
emptier reduces the level-$i$ integer fill by $p$. In order so that
the integer fill increases to $T^{(i)}(t_1)$ by the end of step $t_1$,
it follows that the number of level-$i$ threshold crossings in steps
$t_0, \ldots, t_1$ must be at least
\begin{equation*}
  \begin{split}
    & p(t_1 - t_0 + 1) + T^{(i)}(t_1) - T^{(i)}(t_0 - 1) \\
    & \ge p(t_1 - t_0 + 1) + T^{(i)}(t_1) - d(p - 1)\log n.
  \end{split}
\end{equation*}
\end{proof}

Next we show that, in order for the filler to cross $p(t_1 - t_0 + 1)$
or more level-$i$ thresholds during steps $t_0, \ldots, t_1$, the
filler must cross all of those level-$i$ thresholds in some relatively
small set of cups. Specifically, the set of cups in which the
level-$i$ threshold crossings occur cannot exceed $2 A^{(i)}(t_0 - 1)$
(i.e., twice the number of level-$i$ active cups at step $t_0 - 1$).
\begin{lemma}
Suppose that in steps $t_0, \ldots, t_1$, at least $p(t_1 - t_0 + 1)$
level-$i$ thresholds are crossed. Then the set of cups $S$ in which
those threshold crossings occur must satisfy
$$|S| \le 2 A^{(i)}(t_0 - 1).$$
\label{lem:cupworkingset}
\end{lemma}
\begin{proof}
Suppose that $|S| > 2A^{(i)}(t_0 - 1)$. Let $X$ denote the cups in $S$
that are level-$i$ active at the beginning of step $t_0$, and $Y$
denote the cups in $S$ that are level-$i$ inactive at the beginning of
step $t_0$. Let $s_j$ denote the amount of water placed by the filler
into cup $j$ during steps $t_0, \ldots, t_1$. Define $s_X = \sum_{j
  \in X} s_j$ and $s_Y = \sum_{j \in Y}s_Y$ to be the total water
placed by the filler into cups $X$ and $Y$ during the step interval.

The number of level-$i$ thresholds crossed in each cup $j \in X$
during steps $t_0, \ldots, t_1$ is at most $s_j + 1$. The number of
level-$i$ thresholds crossed in cups $X$ during the steps $t_0,
\ldots, t_1$ is therefore at most
$$\sum_{j \in X} (s_j + 1) = s_X + |X| = s_X + A^{(i)}(t_0 - 1).$$

Since each cup $j \in Y$ is inactive at the beginning of step $t_0 -
1$, but crosses at least one level-$i$ threshold during steps $t_0,
\ldots, t_1$, it must be that $s_j \ge 2$. Moreover, the number of
level-$i$ thresholds crossed in cup $j$ during steps $t_0, \ldots,
t_1$ is at most $s_j - 1$ (because cup $j$ is level-$i$ inactive prior
to step $t_0$). It follows that the number of level-$i$ thresholds
crossed in cups $Y$ during steps $t_0, \ldots, t_1$ is at most,
$$\sum_{j \in Y} (s_j - 1) = s_Y - |Y| \le s_Y - A^{(i)}(t_0 - 1) -
1,$$ where the final inequality follows from the assumption that $|Y|
> A^{(i)}(t_0 - 1)$.

Combining the level-$i$ threshold crossings in $X$ and $Y$, the total
number of crossings is at most
\begin{align*}
&  s_Y - A^{(i)}(t_0 - 1) - 1 + s_Y + A^{(i)}(t_0 - 1) \\ & < s_X + s_Y \le
p(t_1 - t_0 + 1),\end{align*} which completes the proof of the lemma.
\end{proof}

The next lemma establishes a lower bound on the amount of water that
the filler must place in cups $S$ in order for a large number of
level-$i$ threshold crossings to occur in those cups.
\begin{lemma}
Suppose that in steps $t_0, \ldots, t_1$, at least $p(t_1 - t_0 + 1)$
level-$i$ thresholds are crossed in a set of cups $S$. Then at least
$p(t_1 - t_0 + 1) - |S|$ units of water must be placed into cups $S$
during steps $t_0, \ldots, t_1$.
\label{lem:lotsofwaterinS}
\end{lemma}
\begin{proof}
 In order for a cup to cross $k$ level-$i$ thresholds during steps
 $t_0, \ldots, t_1$, at least $k - 1$ units of water must be placed
 into that cup during steps $t_0, \ldots, t_1$. Combining this
 observation for all cups in $S$ completes the proof of the lemma.
\end{proof}

So far we have shown that, in order for the emptier to achieve a large
level-$i$ threshold bolus during steps $t_0, \ldots, t_1$, the filler
must place almost all (i.e., all but $|S|$ units) of their water into
some set $S$ of cups satisfying $|S| \le 2A^{(i)}(t)$. The next lemma
shows that, given the existence of such a set $S$, the level-$i$
threshold bolus will be small as a function of $|S|$ with high
probability in $|S|$. The key ingredient in the proof of the lemma is
the use of the initial random offsets $r_i$ to randomize when
threshold crossings occur in each cup.

\begin{lemma}
Let $m$ satisfy $m \ge 36 \log^{2c} n$ for some $c > 0$. Suppose there
exists a set $S \subset [n]$ of size at most $m$ such that in steps
$t_0, \ldots, t_1$, at least $p(t - t_0 + 1) - m$ units of water are
placed into set $S$. Then with probability at least $1 - \exp(-\log^c
n)$, the number of (any-level) thresholds crossings in steps $t_0,
\ldots, t$ is at most
$$p(t_1 - t_0 + 1) + m^{0.75}.$$
\label{lem:smallStosmallbacklog}
\end{lemma}
\begin{proof}
Say that a cup $j$ \defn{crosses a threshold} (not associated with any
particular level) at a step $t$ if the amount of water $f$ placed by
the filler into cup $j$ during step $t$ satisfies $h_j(t - 1) < s \le
h_j(t - 1) + f$ for some $s \in \mathbb{N}$. Note that the number of
thresholds crossed during each step is independent of which cups the
emptier decides to empty from in preceding steps, since the emptier
always removes integer quantities of water from cups. That is, if
$g_j(t)$ denotes the total amount of water placed by the filler into
cup $j$ in the first $t$ steps (excluding $r_j$), then a threshold is
crossed in cup $j$ during step $t$ if and only if $r_j + g_j(t - 1) <
s \le r_j + g_j(t)$ for some $s \in \mathbb{N}$.
  
For each cup $j$, let $x_j + y_j$ denote the amount of water placed
into cup $j$ during steps $t_0, \ldots, t_1$, where $x_j \in
\mathbb{Z}$ is an integer and $y_j \in [0, 1)$. The first $x_j$ units
  of water that are placed into cup $j$ during steps $t_0, \ldots,
  t_1$ deterministically cross $x_j$ integer thresholds. The final
  $y_j$ units of water placed into cup $j$ then cross an additional
  threshold if and only if
  \begin{equation}
    r_j + g_j(t_0 - 1) < s \le r_j + g_j(t_0 - 1) + y_j,
    \label{eq:crossingcondition}
  \end{equation}
  for some $s \in \mathbb{N}$. Note that \eqref{eq:crossingcondition}
  is equivalent to the condition, $r_j + g_j(t_0 - 1) \pmod 1 \in [1 -
    y_j, 1)$. Since $r_j$ is selected uniformly from $[0, 1)$, it
      follows \eqref{eq:crossingcondition} holds with probability
      exactly $y_j$, meaning that the $y_j$ units of water cross a
      threshold with probability exactly $y_j$.
      
The number of thresholds $L$ crossed during steps $t_0, \ldots, t_1$
can be expressed as
  $$L = \sum_j x_j + \sum_j Y_j,$$ where the $Y_j$'s are independent
zero-one random variables satisfying $\Pr[Y_j = 1] = y_j$.

Since $\E[L] = \sum_j x_j + \sum_j y_j = p(t_1 - t_0 + 1)$, if the number of
thresholds crossings exceeds $p(t_1 - t_0 + 1) + m^{0.75}$, then it must be that
\begin{equation}
  \sum_j Y_j  > \E\left[\sum_j Y_j\right] + m^{0.75}.
  \label{eq:sumyj}
\end{equation}

Recall that there exists a set $S$ of size $m$ such that in steps
$t_0, \ldots, t_1$, at most $m$ units of water are placed in cups $[n]
\setminus S$. It follows that
\begin{equation*}
  \begin{split}
    \E\left[\sum_j Y_j\right] & = \sum_j y_j \\
    & \le m + \sum_{j \in S} y_j \\
    & \le m + |S| \le 2m.
  \end{split}
\end{equation*}

By a multiplicative Chernoff bound, the probability of
\eqref{eq:sumyj} is at most,
\begin{equation*}
  \begin{split}
&  \exp\left(-(m^{-0.25}/2)^2 \cdot (2m) / 3 \right) \\
    & \le \exp\left(-m^{0.5} / 6 \right) \\
    & \le \exp\left(-(\log n)^c\right).
  \end{split}
\end{equation*}
\end{proof}

Combining Lemmas \ref{lem:fillerprogress}, \ref{lem:cupworkingset},
\ref{lem:lotsofwaterinS}, and \ref{lem:smallStosmallbacklog}, we now bound
the probability that, after a given step $t_1$, the number of
level-$(i + 1)$ active cups is large (i.e., $A^{(i + 1)}(t_1) >
2m^{0.75}$ for some $m$) without the number of level-$i$ active cups
having ever been large (i.e., without $A^{(i)}(t_0) > m / 2$ for some
$t_0 \le t_1$).
\begin{proposition}
Let $d$ be a sufficiently large constant and suppose that $c \ge
d$. Let $$m \ge \max(36 (\log n)^{2c}, dp \log n),$$ and let $i \in \{0,
1, 2, \ldots\}$. Then for any step $t_1$, the following statement
holds with probability at least $1 - t_1 \exp\left(-\log^c
n\right)$: Either $A^{(i + 1)}(t_1) \le 2m^{0.75}$, or there is
some step $t_0 \le t_1$ such that $A^{(i)}(t_0) > m / 2$ .
\label{prop:levelanalysis}
\end{proposition}
\begin{proof}
Define $E$ to be the event that $A^{(i + 1)}(t_1) > 2m^{0.75}$ but
$A^{(i)}(t_0) \le m / 2$ for all $t_0 \le t_1$.

Since $T^{(i)}(t_1) \ge A^{(i + 1)}(t_1) > 2m^{0.75}$ in event $E$,
Lemma \ref{lem:fillerprogress} implies that there must be some $t_0
\le t$ such that the number of level-$i$ thresholds $L$ crossed in
steps $t_0, \ldots, t$ satisfies
$$L > p(t_1 - t_0 + 1) + 2m^{0.75} - O(p \log n).$$
For $d$ a sufficiently large constant, it follows that
$$L > p(t_1 - t_0 + 1) + m^{0.75}.$$

By Lemma \ref{lem:cupworkingset}, if event $E$ occurs, then the set
$S$ of cups in which the $L$ level-$i$ thresholds are crossed must
satisfy
$$|S| \le 2 A^{(i)}(t_0 - 1) \le m.$$

In order so that $L \ge p(t_1 - t_0 + 1)$, Lemma
\ref{lem:lotsofwaterinS} requires that the filler places at least
$p(t_1 - t_0 + 1) - |S|$ units of water into the cups $S$ during steps
$t_0, \ldots, t_1$.

So far, we have shown that in order for event $E$ to occur, there must
be some $t_0 \le t_1$ such that more than $p(t_1 - t_0 + 1) + m^{0.75}$
level-$i$ thresholds are crossed in steps $t_0, \ldots, t_1$, but also
such that there exists a set $S$ of size $|S| \le m$ into which the
filler places all but $|S|$ units of water during steps $t_0, \ldots,
t_1$.

For any $t_0$ for which there exists such a set $S$, however, Lemma
\ref{lem:smallStosmallbacklog} bounds the probability of more than
$(t_1 - t_0 + 1) + m^{0.75}$ level-$i$ thresholds being crossed during
the steps by at most $\exp(-\log^c n)$. By a union bound over all $t_0
\le t_1$, it follows that the probability of event $E$ occurring is at
most $t_1 \exp(-(\log n)^c)$.
\end{proof}

To complete the proof of Theorem \ref{thm:randomized}, we use
Proposition \ref{prop:levelanalysis} to show that, with high probability,
there exists some $\ell \le O(\log \log n)$ such that the number of
level-$\ell$ active cups $A^{(\ell)}(t)$ never exceeds $O(p \log n) +
\log^{O(c)} n$ for any $t \le \log^c n$. We then use Theorem
\ref{thm:lognbacklog} to analyze the performance of the greedy
algorithm on level-$\ell$ active cups, thereby bounding the total backlog by
$O(\log p + c \log \log n)$.

\begin{proof}[Proof of Theorem \ref{thm:randomized}]
Let $d$ be a sufficiently large constant and assume that $c \ge d$.
Define $g(1) = n$, and for $i > 1$ define
$$g(i) = \max(4 g(i - 1)^{0.75}, 36 \log^{8c} n, dp \log n).$$

For each step $t_1$ and level $i$, define $X_{t_1, i}$ to be the event
that $A^{(i + 1)}(t_1) > g(i)$ but that for all $t_0 \le t_1$,
$A^{(i)}(t_0) \le g(i - 1)$. By Proposition
\ref{prop:levelanalysis}, applied to $m = 2g(i)$,
$$\Pr[X_{t_1, i}] \le t_1 \exp\left(-\log^{4c}
n\right).$$

Let $\ell \le O(\log \log n)$ be the smallest level $\ell$ such that
$g(\ell) = \max(36 \log ^{8c} n, dp \log n)$. Define $X$ to be the
event that there exists some $i \le \ell$, such that after one of the
steps $t_1 \in \{1, \ldots, \exp(\log^c n)\}$, the number of $i$-level
active cups satisfies $A^{(i)}(t_1) > g(i)$. If event $X$ occurs, and
$i$ is the minimum such $i$, then event $X_{t_1, i - 1}$ must also
occur. (Note that $i \neq 1$, since $g(1) = n$, and thus $A^{(1)}(t) \le
g(1)$ deterministically for all $t$.) Thus we can bound the
probability of event $X$ by,
\begin{equation*}
  \begin{split}
    \Pr[X] & \le \sum_{i = 1}^{\ell} \sum_{t_1 = 1}^{\exp(\log^c n)} t_1 \exp\left(-\log^{4c} n\right) \\
    & \le O\left(\log \log n \cdot \exp(2\log^{c} n) \cdot \exp\left(-\log^{4c} n\right)\right) \\
    & \le \exp\left(-\log^{c} n\right). \\
  \end{split}
\end{equation*}

For the rest of the proof, suppose that event $X$ does not occur. In particular,
\begin{equation}
  A^{(\ell)}(t) \le \max(36 \log^{8c} n, dp \log n) \le p\log^{O(c)} n,
  \label{eq:numactivecups}
\end{equation}
for all $t \le \exp(\log^c n)$ and for some level $\ell \le O(\log
\log n)$. Consider the $\exp(\log^c n)$-step cup game $G$ that occurs
at level $\ell$, in which the fill of each cup $j$ after each step $t$
is the level-$\ell$ fill $h^{\ell}_j(t)$. To complete the proof, we
wish to show that the backlog in $G$ never exceeds $O(\log p + c \log
\log n)$. Since the number of cups containing non-zero water in $G$
never exceeds $p\log^{O(c)} n$ at the end of any step, the number of
cups containing non-zero water must not exceed $p\log^{O(c)} n + p$
at any point in $G$ (including in the intermediate states between steps, when the filler
has inserted water that the emptier has not yet removed). Because the
emptier's greedy algorithm is agnostic to the numbering of cups in
$G$, we may assume without loss of generality that the only cups in
$G$ to ever be non-empty are cups $1, \ldots, p(\log n)^{O(c)} +
p$. Applying the analysis of the greedy algorithm (Theorem
\ref{thm:lognbacklog}) to those cups bounds the backlog by $O(\log p +
c \log \log n)$, as desired.

\end{proof}

\subsection{Lower bounds against randomized emptiers}\label{sec:randomizedlowerbounds}

In this section, we consider two natural ways in which one might hope
to improve Theorem \ref{thm:randomized}, and we prove lower bounds
prohibiting any such improvements.

The first potential improvement to Theorem \ref{thm:randomized} is to
establish a \defn{constant-backlog guarantee}, meaning that when $p$
is large, the emptier can achieve constant backlog $O(1)$ with high
probability in $p$. A constant-backlog guarantee is known to be
possible when the emptier is given resource augmentation
\cite{BenderFaKu19}. Such a guarantee is not possible without the use
of resource augmentation, however. Specifically, Theorem
\ref{thm:randomizedlowerbound1} presents a strategy for the filler
that achieves backlog $\Omega(\log \log p)$ with at least constant
probability.

The second potential improvement to Theorem \ref{thm:randomized} to
establish an \defn{unending guarantee}, meaning that for any
(arbitrarily large) time step $t \in \mathbb{N}$, we present a
high-probability bound on the backlog at time $t$. It was shown by
Bender et al. \cite{BenderFaKu19} that randomizations of the greedy
algorithm can achieve an unending guarantee in the presense of
resource augmentation. Theorem \ref{thm:randomizedlowerbound2} shows
the same type of guarantee cannot be given without the use of resource
augmentation, at least as long as the emptier uses any ``greedy-like''
algorithm, and as long as $p > 1$. Determining whether Theorem
\ref{thm:randomizedlowerbound2} can be generalized to apply to
non-greedy-like algorithms, or to the case of $p = 1$, is an
interesting open question.

We begin by presenting and proving Theorem
\ref{thm:randomizedlowerbound1}, which establishes the impossibility of constant-backlog guarantees.
\begin{theorem}
Suppose $n \ge p + c\log p$ for some constant $c$, and $p \ge 1$. Then
in the $p$-processor cup game on $n$ cups, there exists an oblivious
strategy for the filler and a step-number $s \le \poly(p)$, such that
the backlog at step $s$ is at least $\Omega(\log \log p)$ with
constant probability, regardless of the (possibly randomized) strategy
followed by the emptier.
\label{thm:randomizedlowerbound1}
\end{theorem}

\begin{proof}
The filler performs a series of $p$ phases, each of which consists
of $p^3$ rounds, each of which consists of $c\log p - 1$ steps for
some small constant $c$. At the beginning of each round $r$, the
filler selects an \defn{anchor collection} $A_r$ of $p - 1$
cups. (Note that rounds are numbered starting at $1$ and ending at
$p^{4}$, i.e., the numbering does not restart within each phase.)
The anchor collection begins as $A_1 = \{1, \ldots, p - 1\}$. For most
rounds $r$, the anchor collection $A_{r + 1}$ is selected to be
$A_{r}$. However, for each phase $i$, the filler selects one round $r$
at random from the $p^3$ rounds in the phase, and constructs $A_{r +
  1}$ from $A_{r}$ by selecting a random cup $k \in A_{r}$ and
replacing cup $k$ with a new cup $j \not\in A_{r - 1}$; we describe
the method for selecting cup $j$ shortly. When $A_{r + 1} \neq A_{r}$,
we call round $r$ a \defn{new-anchor round}. Each phase contains
exactly one new-anchor round.

At the beginning of each round $r$, define $B_r$ to be the set of $c
\log p$ smallest-numbered cups not contained in $A_r$. During each of
the $c \log p - 1$ steps in round $r$, the filler places $1$ unit of
water in each cup in $A_r$. Additionally, the filler spreads $1$ unit
of water evenly among the cups in $B_r$, and the filler then removes
one cup at random from the set $B_r$. At the end of round $r$, if
round $r + 1$ is a new-anchor round, then the anchor collection $A_{r
  + 1}$ is selected to be $(A_r \setminus \{k\}) \cup \{j\}$ for a
random $k \in A_r$ and for $j$ the unique cup remaining in $B_r$.

We now prove that the strategy described above for the filler achieves
backlog $\Omega(\log \log p)$ with at least constant probability, at
the end of the construction. Suppose that at the end of some round $r$
in some phase $q$, at least one cup $k$ in $A_r$ contains fill
$\Omega(\log \log p)$ or greater. Then with constant probability at
least
$$\left(1 - \frac{1}{p - 1}\right)^{p - q + 1} \ge \Omega(1),$$ cup $k$ remains in
all of $A_{r + 1}, \ldots, A_{p^{4}}$ (since there are at most $p - q
+ 1$ new-anchor rounds in which cup $k$ could be removed), in which
case the backlog at the end of phase $p$ is at least $\Omega(\log \log
p)$. To complete the proof of the theorem, it suffices to show
that, with at least constant probability, there is some round $r$ at
the end of which at least one cup contains fill $\Omega(\log \log
p)$. We call such a round $r$ a \defn{backlog-enabling round}.

Recall that each phase consists of $p^3$ rounds, each of which
consists of $c \log p - 1$ steps. Call a round $r$
\defn{anchor-preserving} if the emptier removes one unit of water from
each cup in $A_r$ during each step of the round. We claim that for
each phase $q$, the new-anchor round $r$ has probability at least
$\frac{1}{p}$ of either begin backlog-enabling or of being
non-anchor-preserving.  During each step of a round, the emptier can
either remove water from multiple cups in $B_r$, in which case the
round will not be anchor-preserving, or the emptier can remove water
from at most one cup $j$ in $B_r$. In the latter case, since the
filler then removes a random cup from $B_r$, the filler has
probability at least $\frac{1}{|B_r|}$ of removing the cup $j$ that
the emptier emptied from. It follows that with probability at least
$$\frac{1}{c \log p} \cdot \frac{1}{c \log p - 1} \cdot \cdots \cdot
\frac{1}{2} \ge \left(\frac{1}{c \log p}\right)^{c \log p},$$ either
the set $B_r$ never contains any cups that have been emptied out of
during round $r$, or round $r$ is non-anchor-preserving. In the former
case, the single cup $j$ remaining in $B_r$ at the end of round $r$
has fill
$$\frac{1}{c \log p} + \frac{1}{c \log p - 1} + \cdots + \frac{1}{2}
\ge \Omega(\log \log p).$$ Thus there is probability at least $1 / (c
\log p)^{c \log p}$ either that the new-anchor round $r$ is non-anchor
preserving, or that some cup in $A_{r + 1}$ contains at least
$\Omega(\log \log p)$ water at the end of round $r$. For $c$ a
sufficiently small positive constant, $1 / (c \log p)^{c \log p} \ge
\frac{1}{p}$. Thus, with probability at least $\frac{1}{p}$, the
new-anchor round $r$ in a phase $q$ is either non-anchor-preserving or
backlog-enabling.

For a given phase $q$, let $\ell_q$ denote the number of rounds $r$ in
phase $q$ that would be anchor-preserving if the filler were to select
the final round in phase $q$ to be the new-anchor round. Then with
probability $\frac{\ell_q}{p^3}$, the actual new-anchor round $r$
selected by the filler is an anchor-preserving round\footnote{Note, in
  particular, that until the end of the new-anchor round $r$, the
  filler's behavior is indistinguishable from the case where the
  new-anchor round $r$ satisfies $r = p^3$.}. The probability that
new-anchor round $r$ is backlog-enabling is therefore at least
$\frac{1}{p} - \frac{\ell_q}{p^3}$. If $\ell_q \le p^2 / 2$, then the
new-anchor round $r$ is backlog-enabling with probability at least
$\frac{2}{p}$. If, on the other hand, $\ell_q \ge p^2 / 2$, then with
probability at least $\frac{1}{4p}$, the filler selects a new-anchor
round $r$ among the final $p^2 / 4$ rounds in phase $q$;
and thus at least $p^2 / 4$ rounds in phase $q$ are non-anchor
preserving (and all of these rounds occur before the new-anchor round
$r$). But each non-anchor-preserving round increases the amount of
water in cups $A_r$ by at least $1$, meaning that the average fill in
$A_r$ is at least $\Omega(p)$ by the end of the new-anchor round $r$;
and thus round $r$ is backlog-enabling. Thus, if $\ell_q \ge p^2 / 2$,
then there is a backlog-enabling round $r$ in phase $q$ with
probability at least $\frac{1}{4p}$.

Regardless of the value of $\ell_q$, we have shown that there is a
backlog-enabling round in phase $q$ with probability at least
$\Omega(1/p)$. Since there are $p$ phases, some phase will contain a
backlog-enabling round with probability at least $\Omega(1)$. Given
that some phase contains a backlog-enabling round, the backlog at the
end of the construction will be $\Omega(\log \log p)$ with at least
constant probability, completing the proof.
\end{proof}

Call an emptying strategy for the $p$-processor cup game
\defn{$\ell$-greedy-like} if there exists a constant $c$ for which the
following property holds: Whenever there are at least $2$ cups
containing $\ell$ or more water, the emptier will remove water from at
least $2$ cups that contain fill $\ell / c$ or more. That is, the
emptier will not choose $p - 1$ cups that contain less than $\ell / c$
water, if there are at least two cups that contain $\ell$ water or
more.

Theorem \ref{thm:randomizedlowerbound2} establishes that no
greedy-like algorithm can achieve an unending guarantee.

\begin{theorem}
  Let $p \ge 2$. Let $\ell \le n - p$, and suppose $\ell$ is at least
  a sufficiently large constant multiple of $\log p$. Then there
  exists a step $t \le \exp(\ell)$ and an oblivious filling strategy
  that achieves backlog $\Omega(\log \ell)$ after step $t$ with
  probability at least $1 - \exp(-\ell / 5)$, as long as the emptier
  follows a $(-q + \ln \ell)$-greedy-like emptying strategy for some
  sufficiently large constant $q$.
  \label{thm:randomizedlowerbound2}
\end{theorem}
\begin{proof}
Let $c$ be a sufficiently small positive constant, and suppose that
the emptier is $(-1.5 + \ln \frac{\ell}{c})$-greedy-like.
  
The filler's strategy is similar to that in the proof of Theorem
\ref{prop:lowerbound}. The strategy consists of $p \ell \cdot
\exp({\ell / 2}) \le \exp({3\ell/4})$ phases, each of which consists
of $c\ell - 1$ steps. At the beginning of a phase $q$ the filler sets
an anchor-set $A = \{1, \ldots, p - 1\}$ and a non-anchor set $B =
\{p, \ldots, p + c\ell - 1\}$. During the $t$-th step of phase $q$,
the filler places $1$ unit of water into each cup $j \in A$, and
distributes $1$ unit of water evenly among the cups in $B$. The filler
then removes one cup at random from the set $B$.

Call a phase \defn{anchor-preserving} if the emptier removes $1$ unit
from each cup in $A$ during each step of the phase, and \defn{mostly
  anchor-preserving} if the emptier removes $1$ unit from each cup in
$A$ during each of the first $c \ell - 2$ steps of the phase (but not
necessarily during the final step). We claim that if no cups in $A$ contain fill $\Omega(\log \ell)$ at the beginning
of a phase $q$, then phase $q$ is anchor-preserving with probability
at most $1 - 1/(c\ell)^{c \ell}$. Notice that if phase $q$ is mostly
anchor-preserving, then the emptier removes water from at most one cup
in $B$ during each of the steps $1, \ldots, c \ell - 2$ of the
phase. It follows that, with probability at least
$$\frac{1}{c\ell} \cdot \frac{1}{c\ell - 1} \cdot \cdots \cdot
\frac{1}{3} \ge \left(\frac{1}{c\ell}\right)^{c \ell},$$ the phase $q$
is either non-mostly-anchor-preserving (and therefore also not
anchor-preserving), or the two cups $j_1, j_2$ remaining in $B$ prior
to the final step of the phase each contain fill at least
$$\frac{1}{c\ell} + \frac{1}{c\ell - 1} + \cdots + \frac{1}{3} \ge -1.5 + \ln
\frac{\ell}{c}.$$ If cups $j_1, j_2$ both have fill at least
$\Omega(\log \ell)$, then since no cup $j \in A$ has fill $\Omega(\log
\ell)$, the $(- 1.5 + \ln \frac{\ell}{c})$-greedy-like emptier is forced
to remove water from at least two cups in $[n] \setminus A$; this
means that the phase $q$ is non-anchor-preserving. Therefore, if all
the cups in $A$ contain fill sufficiently small in $O(\log \ell)$, then the
probability that phase $q$ is anchor-preserving is at most $1 -
1/(c\ell)^{c \ell}$.

For $c$ a sufficiently small constant, $1/(c\ell)^{c\ell} \ge \exp(-\ell
  / 4)$. Call a sequence of $\exp(\ell / 2)$ phases \defn{successful} if
either some cup in $A$ contains fill $\Omega(\log \ell)$ at the end of
the sequence, or if at least one of the phases is non-anchor
preserving. The probability that a given sequence of $\exp(\ell / 2)$
phases is successful is at least
$$1 - (1 - \exp(-\ell/4))^{\exp(\ell / 2)} \ge 1 - \exp(-\ell / 4).$$

By a union bound, it follows after the first $p \ell \cdot e^{\ell / 2}$
phases (i.e., at the end of the construction), with probability at
least
$$1 - p \ell \exp(\ell / 4) \ge 1 - \frac{1}{e^{\ell / 5}},$$
either some cup in $A$ contains fill $\Omega(\log \ell)$, or
there have been at least $p \log \ell$ non-anchor-preserving
phases. In the latter case, since each non-anchor-preserving phase
increases the total fill in $A$ by at least $1$, some cup must have
fill at least $\Omega(\log \ell)$. This completes the proof of the
theorem.
\end{proof}

\bibliographystyle{abbrv}
\bibliography{all}

\appendix

\section{Comparison of Analysis for $p = 1$ Versus for Larger $p$}\label{sec:appgreedyintuition}

Prior to our work, it was known how to bound the backlog in the
$p$-processor cup emptying game only for $p = 1$ \cite{BenderFaKu19,
  AdlerBeFr03, LitmanMo09}. Our analysis, which applies for $p \ge 1$,
implicitly reduces to (a version) of the known proof when $p$ is taken
to be $1$. In this section, we discuss the known analysis for $p = 1$,
and the relationship to our analysis for $p \ge 1$. We begin by
presenting the analysis for the single-processor cup game:

\begin{lemma}[Adler et al. \cite{AdlerBeFr03}]
The greedy algorithm for the $1$-processor cup game on $n$ cups
achieves backlog $O(\log n)$.
\label{lem:greedyold}
\end{lemma}
\begin{proof}
  We prove that for each step $t$, and for all $k \in [n]$,
  \begin{equation}
    \av_k(S_t) \le \frac{1}{k + 1} + \frac{1}{k + 2} + \cdots + \frac{1}{n}.
    \label{eq:simpleinvariant}
  \end{equation}
  The lemma follows from the case of $k = 1$.

  When $t = 0$, \eqref{eq:simpleinvariant} holds trivially. Suppose by
  induction that \eqref{eq:simpleinvariant} holds for $t - 1$. Let $s$
  be the index of the fullest cup in intermediate state $I_t$. There
  are two cases:

  \noindent \textbf{Case 1: $s$ is among the $k$ fullest cups in
    $S_t$. }In this case, $\tot_k(S_t) = \tot_k(I_t) - 1$. On the
  other hand, since at most one unit of water is added between states
  $S_{t - 1}$ and $I_t$, we have $\tot_k(I_t) \le \tot_k(S_{t - 1}) +
  1$. Thus $\av_k(S_t) \le \av_k(S_{t - 1})$, which by the inductive
  hypothesis satisfies
  $$\av_k(S_{t - 1}) \le \frac{1}{k + 1} + \frac{1}{k + 2} + \cdots + \frac{1}{n}.$$

  \noindent \textbf{Case 2: $s$ is not among the $k$ fullest cups in
    $S_t$. }The average fill of the cups with ranks $2, 3, \ldots, k +
  1$ in $I_t$ is at most as large as the average fill of the cups with
  ranks $1, 2, \ldots, k + 1$ (i.e., $\av_{k + 1}(I_t)$). Since the
  cups with ranks $2, 3, \ldots, k + 1$ in $I_t$ are precisely the $k$
  fullest cups in $S_t$, it follows that
  $$\av_k(S_t) \le \av_{k + 1}(I_t).$$ On the other hand, since at
  most one unit of water is added between states $S_{t - 1}$ and $I_t$,
  \begin{align*}
    \av_{k + 1}(I_t) & \le \av_{k + 1}(S_{t - 1}) + \frac{1}{k + 1} \\ & \le
    \frac{1}{k + 1} + \frac{1}{k + 2} + \cdots + \frac{1}{n},
  \end{align*}
  where
  the final inequality follows by the inductive hypothesis.  
\end{proof}

Extending the proof of Lemma \ref{lem:greedyold} to the $p$-processor
cup game has historically proven difficult due to the following issue,
which we describe in the case of $p = 2$ for simplicity. Suppose that
between states $I_t$ and $S_t$, the rank-2 cup $s$ in state $I_t$
drops to rank greater than $k$. In order to proceed as in Case 2 of
Lemma \ref{lem:greedyold}, we would want to argue that the average
fill of the cups with ranks $1, 3, 4, \ldots, k + 1$ in $I_t$ is at
most the average fill of the cups with ranks $1, 2, 3, \ldots, k +
1$. However, this is only true if $S_2(I_t) \ge \av_{k + 1}(I_t)$. The
critical problem is that, although the rank-1 cup in $I_t$ necessarily
has fill at least $\av_{k + 1}(I_t)$, the same cannot be said for the
cups with ranks $2, 3, \ldots, p$.

In order to try to fix the analysis, instead of proving
\eqref{eq:simpleinvariant}, one could attempt to prove that
\begin{equation}
  \av_k'(S_t) \le \frac{1}{k + 1} + \frac{1}{k + 2} + \cdots + \frac{1}{n},
  \label{eq:simpleinvariantbroken}
\end{equation}
where the \defn{shifted average} $\av_k'(S_t)$ is defined to be the
average fill in cups with ranks $p + 1, \ldots, p + k$. The difficult
case for proving \eqref{eq:simpleinvariantbroken} arises when the
ranks of the cups in $I_t$ are the same as in $S_t$, however. Despite
the fact that the filler may have placed up to $p$ units of water
into the cups with rank $p + 1, \ldots, p + k$ in $I_t$, the emptier
only removes water from the cups with ranks $1, \ldots, p$, allowing
$\av_k'$ to grow by as much as $p / k$.

One of the key insights in Section \ref{sec:greedy} is to analyze the
$N$-bounded skewed average $f^N_k(S_t)$, which combines the positive
aspects of both the standard average and the shifted average. Just as
removing $a$ units of water from the $k$ fullest cups reduces the
standard average $\av_k$ by $a / k$ (assuming the indices of the $k$
fullest cups do not change), removing $a$ units of water from the
$p$ fullest cups reduces the $N$-bounded skewed average $f^N_k$ by $a
/ k$ (assuming the indices of the $k + p$ fullest cups does not
change). On the other hand, just as the $p$ fullest cups in a state
$S$ each have fill at least as large as the shifted average
$\av'_k(S)$, the $p$ fullest cups also have fill at least as large as
the skewed average $f^N_k(S)$, \emph{assuming} that no cup contains
fill more than $N$. Combined, these two properties enable the analysis
of the $N$-truncated cup game given in Section
\ref{sec:Mfinite}. Extending the analysis to non-truncated cup games,
as is done in Section \ref{sec:lemskewaverage}, requires additional
ideas not present in the original analysis for $p = 1$.

The use of the skewed average in our analysis for $p \ge 1$, in place
of the standard average used for $p = 1$, both simplifies the analysis
(as described above), while simultaneously introducing several
combinatorial challenges. Whereas the bound,
$$\av_k(S_t) \le \frac{1}{k + 1} + \frac{1}{k + 2} + \cdots +
\frac{1}{n}$$ immediately implies a bound on the backlog in state $S_t$
(by examining the case of $k = 1$), the same cannot be said for the analogous bound,
$$f^M_k(S_t) \le 1 + \frac{1}{k + 1} + \frac{1}{k + 2} + \cdots +
\frac{1}{n},$$ proven in Section \ref{sec:greedy}. The fact that an
upper bound on $f^M_k(S_t)$ can be used to, in turn, bound $M$, is
somewhat unexpected (to us), since, in general, \emph{larger} values
of $M$ enable \emph{smaller} values of $f^M_k(S_t)$. Nonetheless, by
exploiting the combinatorial structure of the $p$-processor cup game,
and its relationship to the combinatorial definition of $M$, Section
\ref{sec:greedy} is able to achieve such a bound $M$ (and,
consequently, on backlog).

\end{document}